%% file: MaxInfBRkNN.tex
\documentclass[conference]{IEEEtran}
\IEEEoverridecommandlockouts

\usepackage{amsmath}
\usepackage{amssymb}
\usepackage{amsbsy}
\usepackage{graphicx}
\usepackage{flushend}
\usepackage{balance}
\usepackage[linesnumbered,ruled,vlined]{algorithm2e}
\usepackage{algorithmicx}
\usepackage{algpseudocode}

\usepackage{multirow}
\usepackage{booktabs} 

\usepackage{subfigure}
\usepackage{array}
\usepackage{colortbl}
\usepackage{url}
\usepackage{ragged2e}
\usepackage{mathrsfs}

\definecolor{mygray}{gray}{.5}
\definecolor{mypink}{rgb}{.99,.91,.95}
\definecolor{mycyan}{cmyk}{.3,0,0,0}
\usepackage{color}

\newtheorem{example}{Example}
\newtheorem{definition}{Definition}
\newtheorem{lemma}{Lemma}

\newenvironment{proof}{{\it Proof}.}{\hfill $\square$\par}

\newtheorem{observation}{Observation}

\def\BibTeX{{\rm B\kern-.05em{\sc i\kern-.025em b}\kern-.08em
    T\kern-.1667em\lower.7ex\hbox{E}\kern-.125emX}}

\begin{document}

\title{Maximizing the Influence of Bichromatic Reverse $k$ Nearest Neighbors in Geo-Social Networks}


\author{
Pengfei Jin$^{\dagger}$, Lu Chen$^{\dagger}$, Yunjun Gao$^{\dagger}$, Xueqin Chang$^{\dagger}$, Zhanyu Liu$^{\dagger}$, Christian S. Jensen$^{\sharp}$\\
\normalsize $^{\dagger}$\emph{College of Computer Science, Zhejiang University, Hangzhou, China} \\
\normalsize $^{\sharp}$\emph{Department of Computer Science, Aalborg University, Denmark}\\
\emph{\{jinpf, luchen, gaoyj, changxq, zhanyuliu\}@zju.edu.cn \quad csj@cs.aau.dk} \\
}

\maketitle

\begin{abstract}
Geo-social networks offer opportunities for the marketing and promotion of geo-located services. In this setting, we explore
a new problem, called \underline{Max}imizing the \underline{Inf}luence of \underline{B}ichromatic \underline{R}everse \underline{$k$} \underline{N}earest \underline{N}eighbors (MaxInfBR$k$NN). The objective is to find a set of points of interest (POIs), which are geo-textually and socially attractive to social influencers who are expected to largely promote the POIs through online influence propagation. In other words, the problem is to detect an optimal set of POIs with the largest word-of-mouth (WOM) marketing potential. This functionality is useful in various real-life applications, including social advertising, location-based viral marketing, and personalized POI recommendation. However, solving MaxInfBR$k$NN with theoretical guarantees is challenging, because of the prohibitive
overheads on BR$k$NN retrieval in geo-social networks, and the NP
and \#P-hardness in searching the optimal POI set. To achieve practical solutions, we present a framework with carefully designed indexes, efficient batch BR$k$NN processing algorithms, and
alternative POI selection policies that support both approximate and heuristic solutions. Extensive experiments on real and synthetic datasets demonstrate the good performance of our proposed methods.
\end{abstract}



\input{introduction}

\input{relatedwork}

\input{preliminaries}

\input{baseline}

\input{overview}

\input{batch}

\input{poi}
\input{exp}

\input{conclusion}

\balance
\bibliographystyle{abbrv}
\bibliography{ref}
\balance

\end{document}

%% file: introduction.tex
\section{Introduction}
\label{sec:intro}

The problem of bichromatic reverse $k$ nearest neighbor (BR$k$NN) search and maximizing BR$k$NN (MaxBR$k$NN) have received attention due to their importance in a wide range of applications~\cite{choudhury2016vldb,ks2000sigmod,raymond2009vldb}. Given two sets $\mathcal{P}$ and $\mathcal{U}$ representing points of interest (POIs) and users, respectively, the BR$k$NN query is issued for a data point $p\in\mathcal{P}$, to find all users $u\in \mathcal{U}$ having $p$ as one of their $k$ nearest neighbors under a given distance definition~\cite{ks2000sigmod}. The set of BR$k$NN results is also referred as the influence set of $p$, and can be utilized for finding prospective customers~\cite{zhao2017icde}. Based on BR$k$NN, the MaxBR$k$NN query and its variants are proposed to maximize the size of BR$k$NN results via optimal location selection~\cite{raymond2009vldb} or geo-textual tags recommendation~\cite{choudhury2018vldbj,choudhury2016vldb}, and thus target the marketing and promotion of geo-located services.

The increasing popularity of geo-social networks offers more opportunities for location-based marketing. By leveraging the word-of-mouth effect, sellers can advertise sales locations and services online to attract more customers. In social media platforms, influencers, also called cyber-celebrities, play an important role in information dissemination. It is reported that in the USA, more than half of social media users prefer purchasing products or services recommended by the influencers they follow and about 66\% of users believe recommendations by influencers are credible\footnote{\footnotesize https://edition.cnn.com/business/newsfeeds/globenewswire/\\7812666.html}. While influencers and the underlying influence propagation widely exist in the real scenarios, existing BR$k$NN and MaxBR$k$NN studies~\cite{choudhury2016vldb,raymond2009vldb} rarely take this into account. Rather, they assign equal importance to all users in BR$k$NN results, and focus on maximizing the size of BR$k$NNs. The following example illustrates this situation.

\vspace{1mm}
\begin{figure}[t]
\centering
\hspace{-2mm}
\subfigure[Users and POIs]{
 \includegraphics[width=1.08in]{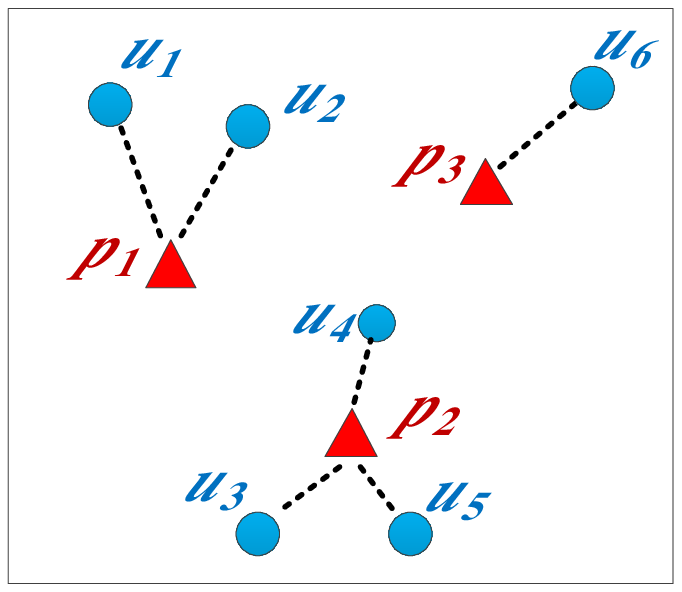}
}\hspace*{0.2mm}
\hspace{3mm}
\subfigure[Social relationships]{
 \includegraphics[width=1.08in]{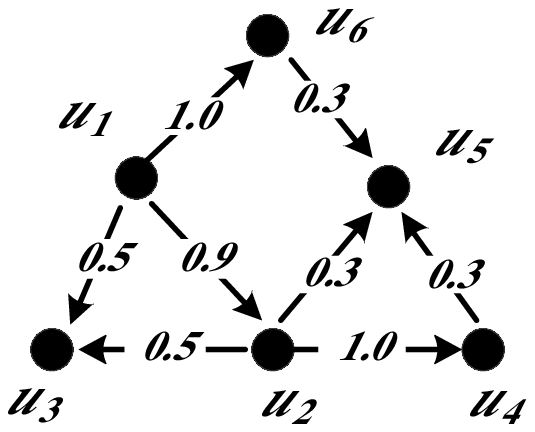}
}
\vspace{-2.5mm}
\caption{A motivating example}
\label{fig:motivating-example}
\vspace{-5mm}
\end{figure}


\begin{example}\label{example:scenarios}
Fig.~\ref{fig:motivating-example}(a) depicts six users (blue circles) and three convenience stores (red triangles), where the dotted line between a user $u$ and a store $p$ means that $u$ has $p$ as one of their $k$ nearest neighbors. Fig~\ref{fig:motivating-example}(b) shows a social network where edges represent the social links between users, with edge labels indicating the probability of one user influencing another user.  Suppose stores $p_1$ and $p_2$ belonging to the same chain, and the manager plans to select the optimal one to launch a marketing campaign. A traditional MaxBR$k$NN would return  $p_2$, because $p_2$ has the largest BR$k$NN result set. However, as shown in Fig.~\ref{fig:motivating-example}(b), the users with $p_2$ as nearest store have little influence and cannot propagate their influence to other users, and thus, the influence is constrained within only three users. To fully enhance the online influence, a better choice may be $p_1$, as users $u_1$ and $u_2$ attracted by $p_1$ embody substantial influence, and they may be able to spread their influence throughout the social network.
\end{example}

We also note that existing related studies \cite{choudhury2018vldbj,choudhury2016vldb,gkorgkas2015sstd,jin2021ickg} aim at enhancing the BR$k$NN results for a single target. Whereas, for sellers with multiple stores or services, a combined promotion for multiple targets is likely to be preferable over a single promotion target.
To this end, we formulate a problem in geo-social networks called \underline{Max}imizing the \underline{Inf}luence of \underline{B}ichromatic \underline{R}everse  \underline{$k$} \underline{N}earest \underline{N}eighbors (MaxInfBR$k$NN). Given a query set $\mathcal{P}_c \subseteq \mathcal{P}$, and two integers $b$ and $k$, MaxInfBR$k$NN aims to find an optimal subset $\mathcal{P}_s$ of $\mathcal{P}_c$, such that $|\mathcal{P}_s|=b$ and the elements in $\mathcal{P}_s$ are highly relevant (i.e., become top-$k$ results) to influencers who can promote POIs effectively online. We assume that users
and POIs are located in a road network, since movement in
real-life settings is constrained to a spatial road network~\cite{zhao2017icde}. When measuring similarities between users and POIs, we consider both geo-textual and social relevance scores, as social media users are more likely to visit nearby relevant places that are also favored by their friends~\cite{gkorgkas2015sstd}. Advertisements in online platforms have normally the limited space or budgets, meaning that $b$ ($=|\mathcal{P}_s|\ll|\mathcal{P}_c|$) should be reasonably small.

Answering MaxInfBR$k$NN queries poses new challenges. A straightforward solution could be to first retrieve BR$k$NN results for each $p\in\mathcal{P}_s$, and then enumerate all the possible size-$b$ POI combinations and return the optimal one whose combined BR$k$NN results has the largest social influence. Although this approach is intuitive, it might be infeasible in practice. The major challenage comes from the prohibitive BR$k$NN retrieval costs in geo-social networks. Existing solutions~\cite{zhao2017icde} tackled this with a time complexity of  $O(\mathcal{|P|}\cdot\mathcal{|U|}^2)$, which fails to scale to larger datasets. Moreover, even though the exact BR$k$NN results are efficiently computed, finding the optimal POI set whose BR$k$NNs have the largest influence is NP-hard, and also inherits \#P-hardness in computing social influence~\cite{kempe2003kdd}. Fortunately, this process can be accelerated by sampling-based techniques~\cite{guo2020sigmod,tang2018sigmod}, but it is still costly and not sufficiently robust across different query inputs (as confirmed by experiments in Section~\ref{sec:exp}).


Motivated by the potential benefits of the MaxInfBR$k$NN query and the lack of effective solutions, this paper makes the following contributions:

\vspace{1mm}

\vspace{0.5mm}
\begin{itemize}\setlength{\itemsep}{-\itemsep}
\item{} We formalize the MaxInfBR$k$NN query/problem in geo-social networks. To the best of our knowledge, this is the first attempt to systematically tackle this problem.
\item{} We prove that the problem is NP-hard, and we present a non-trivial baseline solution with theoretical guarantees.

\item{} We develop efficient batch processing algorithms to retrieve BR$k$NNs for multiple query objects, whose complexity is $O(\xi_1 \cdot\mathcal{|U|}+\xi_2)$ ($\xi_1 (\xi_2) \ll \mathcal{|P|}$). This ensures the scalability of our proposals, and is also a contribution to existing BR$k$NN studies.

\item{} We propose several robust and alternative POI selection policies to answer MaxInfBR$k$NN either approximately or heuristically without excessive sampling costs.

\item{}  We report on a comprehensive empirical study to offer insights into the effectiveness and efficiency of our proposed indexes and algorithms.
\vspace{-0mm}
\end{itemize}

The rest of this paper is organized as follows. Section~\ref{sec:related} reviews the related work. Section~\ref{sec:preli} formalizes the problem. Section~\ref{sec:base} presents a non-trivial baseline solution. Section~\ref{sec:framework_overview} gives a high-level overview of our framework. Section~\ref{sec:batch} covers our indexing schemes and algorithms for batch BR$k$NN processing. Section~\ref{sec:poi} presents alternative POI selection policies to efficiently answer MaxInfBR$k$NN 
queries. We report on an experimental study in Section~\ref{sec:exp}. Finally, Section~\ref{sec:conclu} concludes the paper, and provides future work suggestions.

\vspace{1.0mm}

%% file: relatedwork.tex
\vspace{-2mm}
\section{Related Work}
\label{sec:related}
This section briefly overviews the existing studies related to our problem, including BR$k$NN and MaxBR$k$NN queries, influence maximization, and geo-social keyword queries.

\vspace{0mm}

\subsection{BR$k$NN and MaxBR$k$NN Queries}

The bichromatic reverse $k$ nearest neighbor (BR$k$NN) query was first studied by Korn et al. \cite{ks2000sigmod}. Given two sets $\mathcal{U}$ and $\mathcal{P}$ of data points, and an object $p\in\mathcal{P}$, the query finds all the points in $\mathcal{U}$ that have $p$ as one of their $k$ nearest neighbors in $\mathcal{P}$. Many variants of this query have been studied, such as R$k$NN in large graphs \cite{yiu2006tkde}, the Reverse Spatial and Textual $k$ Nearest Neighbor (RST$k$NN)  \cite{lu2011sigmod,lu2014acmtrans}, and Reverse Top-$k$ Geo-Social Keyword Query (R$k$GSKQ) \cite{zhao2017icde}. Recently, inspired by profile-based marketing~\cite{raymond2009vldb}, the MaxBR$k$NN query that maximizes the result size of BR$k$NNs was studies. Many variants exist, e.g., MaxBR$k$NN for trajectories \cite{rahat2018adc}, MaxBRST$k$NN for geo-textual data \cite{choudhury2016vldb,gkorgkas2015sstd}, and MaxBR$k$NN for streaming geo-data~\cite{luo2018dasfaa}, to name just a few. These studies all assign equal importance to the data points in BR$k$NN, and focus on maximizing the result size. However, in real scenarios, different data points in BR$k$NNs may have different influence. This motivates our problem. Huang et al.~\cite{hung2014tsas} consider social influence in BR$k$NN, but only find a single optimal location, while we find an optimal set of multiple geo-social targets, which is NP-hard. Moreover, they define the similarity by the Euclidean distance, whereas our similarity notion considers road network distance and social relevance, which is more complex and realistic.

\vspace{-1mm}
\subsection{Influence Maximization}
Influence maximization (IM) was first studied by Domingos and Richardson \cite{domingos2001kdd,richardson2002kdd}. It is to find a set of users called seeds to trigger the largest expected influence propagation in social networks. Kempe et al.~\cite{kempe2003kdd} formulated it as a discrete optimization problem, proved that it is NP-hard~\cite{kempe2003kdd}, and proposed an algorithm with a $(1-1/e)$-approximation ratio by Monte-Carlo simulation. Since then, substantial research has been devoted to developing more efficient and scalable IM algorithms~\cite{chen2010kdd,cheng2014sigir,guo2020sigmod,tang2018sigmod,tang2017asonam,tang2015sigmod}. Among them, the reverse influence sampling (RIS) techniques (first presented in~\cite{borgs2014soda} and then optimized by~\cite{guo2020sigmod,tang2018sigmod,tang2015sigmod}) are widely considered as the state-of-the-arts in  addressing IM problem with theoretical guarantees. Recently, RIS has also been extended to address many IM variants~\cite{bian2020vldb,huang2020vldbj,wang2017tkde}. Although IM and many of its variants select seed users freely and only return identical results when the seed set size is fixed, our problem selects target POIs, so as to attract more influencers for online propagation. Thus, our problem is query-dependent and more flexible in personalized applications.

\vspace{-1mm}
\subsection{Geo-Social Keyword Queries}
Geo-Social Keyword Queries (GSKQ) consider geographical, textual, and social similarities when quantifying the relevance of objects. Existing studies include i) the Top-$k$ Geo-Social Keyword Query (T$k$GSKQ) \cite{ahuja2015sstd} that finds $k$ most relevant objects for a query user; ii) the Reverse Top-$k$ Geo-Social Keyword Query (R$k$GSKQ) \cite{jin2020dasfaa,zhao2017icde} that detects potential customers in social networks; and iii) the Why-not Questions on Top-$k$ Geo-social Keyword Query (WNGSKQ)~\cite{zhao2018icde} that is designed to refine and improve the query results. Many variants of geo-social queries also exist, such as social-aware top-$k$ spatial keyword search~\cite{wu2012mdm}, geo-social ranking~\cite{armenatzoglou2015vldbj}, and community search in geo-social networks~\cite{guo2019icde,guo2021icde}, to name but a few. All these studies aim at different problems, and hence, their solutions do not handle our problem.

%% file: preliminaries.tex
\begin{table}[t]\small
\centering
\setlength{\tabcolsep}{2pt}
\caption{Symbols and description}
\label{tab:symbol}
\begin{tabular}{|p{1.5cm}|p{7.0cm}|}
\hline
\textbf{Notation} & \textbf{Description} \\ \hline
$\mathcal{G}_r$($\mathcal{G}_s$)  & a spatial road (social) network \\ \hline
$\mathcal{P}$($\mathcal{U}$)  & a set of POIs (users) located on $\mathcal{G}_r$\\ \hline
$dist(p,u)$ & the shortest path distance from a POI $p$ to a user $u$ \\ \hline
$F_{GST}(u,p)$ & the geo-social and textual similarity score of $p$ to $u$ \\ \hline
$\mathcal{S}_{t}(u)$  & the top-$k$ ($k$NN) result set of $u$  \\ \hline
$\mathcal{S}_{r}(p)$  & the  bichromatic reverse $k$NN result set of $p$ (the potential users of $p$) \\ \hline
$\mathcal{P}_{c}$  & the candidate POI set in the query inputs\\ \hline
$\mathcal{P}_{s}$  & the POI set selected for market promotion\\ \hline
$\mathcal{L}_{\textit{SB}}^{\downarrow}(u)$  & the scoring lower bound list of a user $u$  \\ \hline
$\mathcal{M}_{\textit{SB}}^{\downarrow}(u_{ps})$  & the scoring lower bound map of a pseudo user $u_{ps}$  \\ \hline

$\mathcal{P}_{b}^{*}$  & a size-$b$ POI set retrieved by the greedy algorithm\\ \hline

$\mathcal{P}_{b}^{\textit{opt}}$  & the optimal size-$b$ subset selected from $\mathcal{P}_{c}$\\ \hline

$\mathbb{I}_{p}(\mathcal{P}_{i})$  & the influence of a POI set $\mathcal{P}_{i}$ \\ \hline

\end{tabular}
\label{tab:symb}
\vspace{-0.15in}
\end{table}

\section{Preliminaries}
\label{sec:preli}
Table~\ref{tab:symbol} lists the frequently used notation. In our problem, a geo-social network is a heterogeneous network structure as shown in Fig.~\ref{fig:LBSN-dataset-example}(a). Specifically, a geo-social network is composed by a road network $\mathcal{G}_{r}=(\mathcal{V}_r, \mathcal{E}_r)$ and a social network $\mathcal{G}_{s}=(\mathcal{V}_s, \mathcal{E}_s)$, where $\mathcal{V}_r$ ($\mathcal{V}_s$) and $\mathcal{E}_r$ ($\mathcal{E}_s$) denote the vertex set and the edge set respectively. Two types of geo-textually and socially tagged objects are located in $\mathcal{G}_{r}$, i.e., a set $\mathcal{P}$ of POIs and a set $\mathcal{U}$ of users that have reviews on POIs and correspond to the nodes in $\mathcal{G}_{r}$. Each POI $p\in \mathcal{P}$ is a triple $\left(loc, \mathit{key}, \mathcal{CK}\right)$, where $p.loc=( v_{\textit{NN}},dis )$ is a geo-spatial position in $\mathcal{G}_{r}$, given by the nearest vertex  $v_{\textit{NN}}$ of $p$ in $\mathcal{V}_r$ and a distance $dis$ to that vertex; $p.key$ is a set of weighted  keywords describing $p$; $p.\mathcal{CK}$ is a set of IDs of users who have checked into $p$. Each user $u \in \mathcal{U}$ is denoted by a triple $\left(loc, key, F(u)\right)$, where $u.loc$ is the location descriptor similar as $p.loc$, $u.key$ is a set of keywords capturing the users' interests, and $F(u)$ is the friend set of $u$ in $\mathcal{G}_{s}$. The detailed information of POIs and users in Fig.~\ref{fig:LBSN-dataset-example}(a) is depicted in Figs.~\ref{fig:LBSN-dataset-example}(c) and~\ref{fig:LBSN-dataset-example}(d). Note that users $u_9$ to $u_{12}$ in Fig.~\ref{fig:LBSN-dataset-example}(a) have empty keyword sets and thus are not shown in Fig.~\ref{fig:LBSN-dataset-example}(d).

\begin{figure*}[t]
\centering
\raisebox{-2cm}{ \includegraphics[width=6.9in]{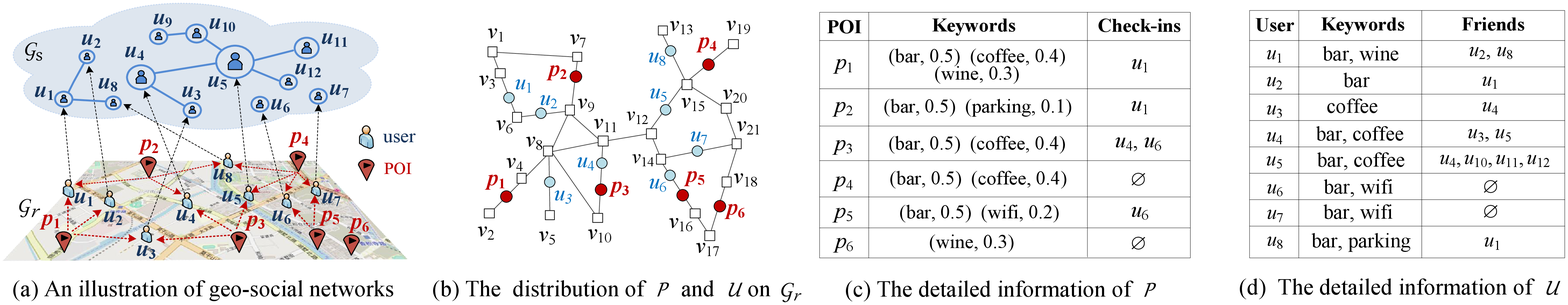}}
\vspace*{-1.5mm}
\caption{{Example of MaxInfBR$k$NN problem setting}}
\label{fig:LBSN-dataset-example}
\vspace*{-4mm}
\end{figure*}

\vspace{1mm}
\begin{definition}\label{defn:tkgsk}
{\bf \textit{(Top-$k$ Geo-Social Keyword  (T$k$GSK) queries)}}. Given a POI set $\mathcal{P}$ and a user $u$, a T$k$GSK query is issued by $u$ to find a set $S_t(u)$ with $k$ POIs $p$ in $\mathcal{P}$, which are most relevant to $u$ based on the scoring function as below:
\end{definition}

\vspace{-3mm}
\begin{small}
\begin{equation}
\small
\label{equ:gsk}
\hspace{8mm} F_{GST}(u,p)= \frac{\alpha \cdot f_s (u,p)+(1-\alpha)\cdot f_t(u,p)}{f_g(u,p)}
\end{equation}
\end{small}

\noindent In Eq.~\ref{equ:gsk}, $f_{g}(u,p)$ denotes the geographical proximity measured by the shortest path distance between $u$ and $p$ in $\mathcal{G}_r$, $f_{t}(u,p)$ represents the textual similarity computed by the TF-IDF metric \cite{choudhury2018vldbj,salton1988inf}, and $f_{s}(u,p)$ is the social relevance whose formulation follows the literature~\cite{zhao2018icde,zhao2017icde}. The parameter $\alpha\in$ [0, 1] balances the importance of social relevance and textual similarity. Here, a large score of $p$ means that $p$ has high relevance to $u$. Note that Eq.~$1$ is formulated as a ratio rather than as a linear function which is commonly used in extensive studies \cite{ahuja2015sstd,zhao2017icde}. This is done to avoid expensive normalization of the geographical score with prior knowledge~\cite{rocha2012edbt} (i.e., the largest shortest path distance). 
Nonetheless, our methods can be easily extended to support the linear function as well.

\vspace{1mm}
\begin{definition}
\label{dfn:rkgsk}
{\bf \textit{(Bichromatic Reverse $k$ Nearest Neighbor (BR$k$NN) in the geo-social network)}}. Given a set $\mathcal{P}$ of POIs and a set $\mathcal{U}$ of users, in a geo-social network, a BR$k$NN query is issued for a POI $p \in \mathcal{P}$, and returns all the users $u \in \mathcal{U}$ having $p$ in their T$k$GSK query results. Thus, if the result is denoted as $\mathcal{S}_r(p)$, we have $\mathcal{S}_r(p) = \{u | u\in \mathcal{U} \wedge p \in S_t(u)\}$.
\end{definition}

\begin{example}
\vspace{0.5mm}
\label{examp-reverse}
In Fig.~\ref{fig:LBSN-dataset-example}(a), the T2GSK query result set $S_t(u_1)$ for $u_1$ is $\{p_1, p_2\}$, as $p_1$ and $p_2$ are most relevant to $u_1$ among all $p\in\mathcal{P}$ according to Eq.~\ref{equ:gsk}. Similarly, $S_t(u_2)$ = $\{p_1,p_2\}$ and $S_t(u_3)$ = $\{p_1,p_3\}$. Thus,  $S_r(p_1)$ = $\{u_1,u_2,u_3\}$.
\end{example}

\vspace{0.5mm}
For a POI set $\mathcal{P}_{i}$, we denote $\mathcal{S}_{r}(\mathcal{P}_i)$ as the union of BR$k$NN results for all $p \in \mathcal{P}_{i}$, i.e., $\mathcal{S}_r(\mathcal{P}_i) = \bigcup_{p\in \mathcal{P}_i}S_r(p)$. As the POIs in $\mathcal{P}_i$ are highly relevant to the users in $\mathcal{S}_{r}(\mathcal{P}_i)$, these users tend to be interested in $\mathcal{P}_i$, and become the potential customers. For simplicity,  we also refer to the BR$k$NNs as potential users. We next introduce the Independent Cascade (IC) model~\cite{kempe2003kdd}, a classic and dominant information diffusion model, to capture the influence of users in BR$k$NNs.

\vspace{1mm}
{\bf IC Model}. Given a social network $\mathcal{G}_s=(\mathcal{V}_s,\mathcal{E}_s)$, the IC model assigns a weight $w(u,v)\in[0,1]$ to each edge $(u,v)\in\mathcal{E}_s$ that denotes the probability that $v$ can be influenced by $u$. The influence propagation over $\mathcal{G}_s$ is modeled as an iterative stochastic process. Initially, a set of users called seeds (denoted as $\mathcal{S}$) are influenced. Then, seed users further spread their influence through the social network following randomized rules. Specifically, in each iteration, when a user $u$ is newly influenced/activated, the user has a single chance to activate each of inactive friend $v$ with  probability $w(u,v)$. Any influenced user remains in this state until the end. The propagation process proceeds until no additional users in $\mathcal{G}_s$ can be further influenced. Let $I(\mathcal{S})$ be the number of all influenced users when the propagation stops. Due to the randomness in the propagation process, the influence $\mathbb{I}(\mathcal{S})$ of $\mathcal{S}$ is evaluated by the expected number of influenced/activated users over all the possible propagation instances i.e., $\mathbb{I}(\mathcal{S})= E(I(\mathcal{S}))$.

\vspace{0.5mm}
\begin{definition}
\label{dfn:MaxInf}
{(\bf \textit{MaxInfBR$k$NN})}. Given a query set $\mathcal{P}_c$ ($\mathcal{P}_c\subset \mathcal{P})$, and an integer $b$ $(b\leq |\mathcal{P}_c|)$, the MaxInfBR$k$NN problem is to find a set $\mathcal{P}_{s} \subset \mathcal{P}_c$, such that $|\mathcal{P}_{s}|=b$ and the influence of $\mathcal{P}_{s}$ is the largest over all size-$b$ subsets of $\mathcal{P}_{c}$. Here, the influence of $\mathcal{P}_{s}$ $($i.e., $\mathbb{I}_p(\mathcal{P}_{s}))$ is given by $\mathbb{I}(\mathcal{S}_r(\mathcal{P}_{s}))$, which is the expected influence spread when users in $\mathcal{S}_{r}(\mathcal{P}_{s})$ tend to seed the influence propagation under the IC model. Formally, $\mathcal{P}_{s} =\arg \max_{\mathcal{P}_i \subseteq \mathcal{P}_c,|\mathcal{P}_i|= b}  \mathbb{I}_p(\mathcal{P}_i)$.
\end{definition}

\begin{example}
\vspace{-1mm}
\label{examp:maxinf}
In Fig.~\ref{fig:LBSN-dataset-example}, assume that a company owns three stores $p_1$, $p_3$, and $p_5$. When $k$ = 2, the potential users of each store are $S_r(p_1)=\{u_1,u_2,u_3\}, S_r(p_3)=\{u_3,u_{4},u_5\}$, and $S_r(p_5)=\{u_6,u_7\}$. For simplicity, we set the weight of each edge in $\mathcal{G}_s$ to 1. Given $b=2$, we obtain the optimal POI set $\mathcal{P}_s= \{p_1,p_3\}$ whose influence is 10, which is the total number of influenced users in $\mathcal{G}_s$, including potential users $($i.e., $u_1, u_2, u_3, u_4,$ and $u_5)$ of $\{p_1,p_3\}$ and the users $($i.e., $u_8,u_9,u_{10},u_{11},$ and $u_{12})$ further activated by the potential users via influence propagation.

\end{example}

\vspace{1mm}
\begin{lemma}
\label{np-hard}
MaxInfBR$k$NN is NP-hard.
\end{lemma}
\vspace{0mm}

\begin{proof}
The problem of MaxInfBR$k$NN can be reduced from a well-known NP-hard problem, the Maximum Coverage (MC) problem~\cite{feige1998acmj}. Given a collection of subsets $\mathcal{S}=\{S_1,S_2,...,S_m\}$ of a ground set $\mathbb{G}=\{e_1,e_2,...,e_n\}$, and a positive integer $b$, the MC problem aims to find  $\mathcal{S}'\subset\mathcal{S}$, such that $|\mathcal{S}'|=b$ and $\mathcal{S}'$ covers the largest number of distint elements in $\mathbb{G}$. The problem of MaxInfBR$k$NN is a generalization of MC. If each user in $S_r(\mathcal{P}_c)$ has no social links in $\mathcal{G}_s$ (a.k.a. cold start users), then the problem of solving MaxInfBR$k$NN becomes equivalent to solving the MC problem, with $S_r(p)$ of each $p\in\mathcal{P}_c$ corresponding to each subset in $\mathcal{S}$ defined by the MC problem. Therefore, the problem of MaxInfBR$k$NN is at least as hard as the MC problem, and thus MaxInfBR$k$NN is NP-hard. \end{proof}

%% file: baseline.tex
\section{Baseline Solution}
\label{sec:base}

 MaxInfBR$k$NN is NP-hard, and also inherits \#P-hardness in exactly computing social influence~\cite{chen2010kdd,kempe2003kdd}. To avoid brute force search, we can theoretically approximate the optimal solution based on the following lemma.

\begin{lemma}
\label{submodular}
The objective function $\mathbb{I}_{p}(\cdot)$ evaluating the influence of a set of POIs satisfies the following properties:

\vspace{0.5mm}
\textbf{P.1} {\rm (\textbf{Non-negative})} For any non-empty POI set $\mathcal{P}_i$, $\mathbb{I}_p(\mathcal{P}_i) \geq 0$, and for an empty POI set, $\mathbb{I}_p(\varnothing)=0$.

\textbf{P.2} {\rm (\textbf{Monotonic})} For any POI set $\mathcal{P}_i$ and any POI $p$ from the candidate POI set $\mathcal{P}_c$, $\mathbb{I}_p(\mathcal{P}_i) \leq \mathbb{I}_p(\mathcal{P}_i \cup p)$.

\textbf{P.3} {\rm (\textbf{Submodular})} {Given two POI sets $\mathcal{P}_i$ and $\mathcal{P}_j$ with $\mathcal{P}_i\subset\mathcal{P}_j\subset\mathcal{P}_c$, for each POI $p\in\mathcal{P}_c\backslash\mathcal{P}_j$}, we always have $\left(\mathbb{I}_p(\mathcal{P}_j\cup p)-\mathbb{I}_p(\mathcal{P}_j)\right) \leq \left(\mathbb{I}_p(\mathcal{P}_i\cup p)-\mathbb{I}_p(\mathcal{P}_i)\right)$.
\end{lemma}

\vspace{1mm}
\begin{proof}
P.1 and P.2 are obvious according to the definition of $\mathbb{I}_p(\cdot)$. Hence, we only focus on proving P.3.
In the IC model, the influence evaluation can also be measured by possible world semantics \cite{kempe2003kdd}. Specifically, let $\mathbb{G}_I$ denote a set of graph instances $\{\mathcal{G}_s^{1},\mathcal{G}_s^{2},...,\mathcal{G}_s^{n}\}$, where $\mathcal{G}_s^{t}=(\mathcal{V}_s,\mathcal{E}_{s}^{t})$ ($1\le t\le n$, $\mathcal{E}_{s}^{t}\subseteq\mathcal{E}_{s}$) is a possible world (i.e., a propagation instance), and $\mathcal{E}_{s}^{t}$ is independently sampled from $\mathcal{E}_{s}$. Then, the existence probability of $\mathcal{G}_{s}^{t}$ (i.e., $Pr(\mathcal{G}_s^{t})$ ) is computed as:
\vspace{1mm}
\begin{small}
\begin{equation}
\label{equa:possible}
    \operatorname{Pr}(\mathcal{G}_s^{t})=\prod_{(u,v) \in \mathcal{E}_{s}^{t}} w(u,v) \prod_{(u',v') \in \mathcal{E}_{s} \backslash \mathcal{E}_{s}^{t}}(1-w(u',v'))
\end{equation}
\end{small}
\vspace{1mm}
{\hspace{-2mm} Let $R_{\mathcal{G}_s^{t}}(\mathcal{P}_i,u)$ be an indicator. If a user $u$ can be reached by at least one user from $\mathcal{S}_r(\mathcal{P}_i)$ in $\mathcal{G}_s^{t}$, we set $R_{\mathcal{G}_s^{t}}(\mathcal{P}_i,u)=1$; otherwise, $R_{\mathcal{G}_{s}^{t}}(\mathcal{P}_i,u)=0$. Therefore, the influence of POIs $\mathbb{I}_p(\mathcal{P}_i)$ can be calculated as follows:}
\begin{small}
\begin{equation}
\label{equa:inf_possible}
\vspace{0mm}
\begin{small}
\mathbb{I}_p(\mathcal{P}_i)=\mathbb{I}(\mathcal{S}_r(\mathcal{P}_i))=\sum_{\mathcal{G}_{s}^{t} \in \mathbb{G}_I } \sum_{u \in \mathcal{V}_{s}} R_{\mathcal{G}_{s}^{t}}(\mathcal{P}_i, u)  \operatorname{Pr}\left(\mathcal{G}_{s}^{t}\right)
\end{small}
\end{equation}
\end{small}
\hspace{-2.8mm} Given a POI $p\in\mathcal{P}_c$ and a possible world $\mathcal{G}_{s}^{t}$, if there exists at least one user $v$ in $\mathcal{S}_r(p)$ that can reach $u$ in $\mathcal{G}_{s}^{t}$,  $R_{\mathcal{G}_s^{t}}(\mathcal{P}_i\cup\{p\},u)$=$R_{\mathcal{G}_s^{t}}(\mathcal{P}_j\cup\{p\},u)$=1. Thus, we have ($R_ {\mathcal{G}_s^{t}}(\mathcal{P}_j\cup\{p\},u)-R_{\mathcal{G}_s^{t}}(\mathcal{P}_j,u))-(R_{\mathcal{G}_s^{t}}(\mathcal{P}_i\cup\{p\},u)-R_{\mathcal{G}_s^{t}}(\mathcal{P}_i,u))  = (R_{\mathcal{G}_s^{t}}(\mathcal{P}_i,u)-R_{\mathcal{G}_s^{t}}(\mathcal{P}_j,u)) \leq 0$.

Otherwise, the influence of $p$ cannot reach $u$, and we further have  $R_{\mathcal{G}_s^{t}}(\mathcal{P}_j\cup\{p\},u)-R_{\mathcal{G}_s^{t}}(\mathcal{P}_j,u)-(R_{\mathcal{G}_s^{t}}(\mathcal{P}_i\cup\{p\},u)-R_{\mathcal{G}_s^{t}}(\mathcal{P}_i,u))=0$. Hence, we can conclude that ($R_{\mathcal{G}_s^{t}}(\mathcal{P}_j\cup\{p\},u) - R_{\mathcal{G}_s^{t}}(\mathcal{P}_j,u)) \leq (R_{\mathcal{G}_s^{t}}(\mathcal{P}_i\cup\{p\},u)-R_{\mathcal{G}_s^{t}}(\mathcal{P}_i,u))$.

Based on this, we can derive that ($\mathbb{I}_{p}\left(\mathcal{P}_j\cup\{p\}\right)-\mathbb{I}_{p}\left(\mathcal{P}_j \right)) \\ \le (\mathbb{I}_p(\mathcal{P}_i\cup\{p\})-\mathbb{I}_p(\mathcal{P}_i))$ according to Eq.~\ref{equa:inf_possible}, which proves P.3. The proof completes. \end{proof}

As is proved in \cite{feige1998acmj}, submodular optimization problems can be approximated by a ratio no worse than $1-1/e\approx0.632$. Thus we present a non-trivial baseline (\textbf{denoted as BA}) by extending state-of-the-art techniques. The workflow of BA includes three major procedures:

\textbf{BR$k$NNs retrieval.} First, we perform the state-of-the-art GIM-Tree based algorithm~\cite{jin2020dasfaa,zhao2017icde} to retrieve the exact BR$k$NNs result for each $p\in\mathcal{P}_c$. Then, we add a directed edge with weight 1 from each POI $p\in\mathcal{P}_{c}$ to a user $u$ if $u\in\mathcal{S}_r(p)$. This yields a heterogeneous graph $\mathcal{G}_{H}$, whose nodes contain both users and POIs. To avoid unnecessary computations, we remove all users in $\mathcal{G}_{H}$ that cannot reached by any POI in $\mathcal{P}_c$.

\textbf{Reverse influence sampling.}
Second, we extend the state-of-the-art reverse influence sampling (RIS) technique~\cite{guo2020sigmod, tang2018sigmod} to support POI selection. The core idea of RIS is to generate random reverse reachable (RR) sets \cite{borgs2014soda} to estimate the social influence. However, our RR set generation differs to previous methods~\cite{guo2020sigmod, tang2018sigmod}, since we only generate RR sets particularly for POIs. More specifically, we implement RR set generation by uniformly sampling a user $u_i$ from $\mathcal{G}_{H}$ and performing stochastic reverse breadth first search (BFS) from $u_i$~\cite{tang2015sigmod}. The difference is that our BFS is conducted on $\mathcal{G}_{H}$ rather than $\mathcal{G}_{s}$, and only adds POIs in $\mathcal{P}_c$ reached by $u_i$ into the RR set. More optimizations for RR set generation are available in~\cite{guo2020sigmod}. Following previous studies~\cite{guo2020sigmod, tang2018sigmod}, we generate two collections of RR sets $\mathcal{R}_{1}$ and $\mathcal{R}_{2}$ in a sequential manner, and utilize $\mathcal{R}_{1}$ and $\mathcal{R}_{2}$ to determine selected POIs.

\textbf{POI selection and qualification.} 
Whenever $|\mathcal{R}_{1}|=|\mathcal{R}_{2}|=2^{i}$, we apply the greedy algorithm to select $b$ POIs from $\mathcal{P}_c$ covering the largest number of RR sets in $\mathcal{R}_{1}$. Let $\mathcal{P}_b^{*}$ be the POI set selected by the greedy algorithm. Then, we utilize $\mathcal{R}_{2}$ to determine an influence lower bound $\mathbb{I}^{-}_{p}(\mathcal{P}_b^{*})$ on $\mathbb{I}_{p}(\mathcal{P}_b^{*})$. Let $\mathcal{P}_b^{opt}$ denote the POI set with the largest influence among all size-$b$ subsets of $\mathcal{P}_c$. Next, we further use $\mathcal{R}_{1}$ to derive an influence upper bound $\mathbb{I}^{+}_{p}(\mathcal{P}_b^{opt})$ on $\mathbb{I}_{p}(\mathcal{P}_b^{opt})$. According to $\mathbb{I}^{-}_{p}(\mathcal{P}_b^{*})$ and $\mathbb{I}^{+}_{p}(\mathcal{P}_b^{opt})$, the worst approximation ratio of $\mathcal{P}_b^{*}$ is computed by $\mathbb{I}^{-}_{p}(\mathcal{P}_b^{*})/\mathbb{I}^{+}_{p}(\mathcal{P}_b^{opt})$. Given an error threshold $\epsilon$, if the derived worst approximation ratio is no less than $1-1/e-\epsilon$, $\mathcal{P}_b^{*}$ is returned. Otherwise, we repeat procedures 2 and 3 until we obtain a qualified $\mathcal{P}_b^{*}$. The bounds formulation of both $\mathbb{I}^{-}_{p}(\mathcal{P}_b^{*})$ and $\mathbb{I}^{+}_{p}(\mathcal{P}_b^{opt})$ are covered in Lemma~\ref{lemma:poi-inf-upper-lower-inf}.

\begin{lemma}
 \label{lemma:poi-inf-upper-lower-inf}
Given $\mathcal{R}_1$, $\mathcal{R}_2$, $\mathcal{P}_{b}^{*}$, and $\mathcal{P}_{b}^{opt}$ defined as above, let $\Lambda_{\mathcal{R}_{2}}(\mathcal{P}_{b}^{*})$ be the number of RR sets covered by $\mathcal{P}_{b}^{*}$ in $\mathcal{R}_{2}$ (the coverage of $\mathcal{P}_{b}^{*}$), and $\Lambda^{u}_{\mathcal{R}_{1}}(\mathcal{P}_{b}^{opt})$ be the upper bound of the coverage of $\mathcal{P}_{b}^{opt}$ in $\mathcal{R}_{1}$. Then, we have \\
\begin{small}
\begin{equation}
\label{equa:inf_lower}
\mathbb{I}^{-}_{p}\left(\mathcal{P}_{b}^{*}\right)=\left(\left(\sqrt{\Lambda_{\mathcal{R}_{2}}\left(\mathcal{P}_{b}^{*}\right)+\frac{2 \eta _{l}}{9}}-\sqrt{\frac{\eta_{l}}{2}}\right)^{2}-\frac{\eta_{l}}{18}\right) \cdot \frac{|\mathcal{V}_s|}{|\mathcal{R}_{2}|}
\end{equation}
\begin{equation}
\label{equa:inf_upper}
\mathbb{I}^{+}_{p}\left(\mathcal{P}_{b}^{opt}\right)=\left(\sqrt{\Lambda_{\mathcal{R}_{1}}^{u}\left(\mathcal{P}_{b}^{opt}\right)+\frac{\eta _{u}}{2}}+\sqrt{\frac{\eta _{u}}{2}}\right)^{2} \cdot \frac{|\mathcal{V}_s|}{|\mathcal{R}_{1}|}
\end{equation}
\end{small}

\noindent where $\eta_{l}$=$\ln(1/\delta_{l})$ {\rm(}$\eta_{u}$=$\ln(1/\delta_{u})${\rm)}, $\delta_{l}$ {\rm(}$\delta_{u}${\rm)} is the  probability when the lower {\rm(}upper{\rm)} bound estimation in Eq.~\ref{equa:inf_lower} {\rm(}Eq.~\ref{equa:inf_upper}{\rm)} fails. In ~\cite{guo2020sigmod,tang2018sigmod},  $\delta_{l}$=$\delta_{u}$=$1/2\cdot\delta$, with $\delta$ as the total error probability whose value is tunable according to users' requirements.

\end{lemma}  

\begin{proof}
The lemma is extended from Lemmas $4.2$ and $4.3$ by Tang et al.~\cite{tang2018sigmod} and thus the detailed proof is omitted. \end{proof}

\begin{figure}[t]
\centering
\hspace{-2mm}
\includegraphics[width=3.0in]{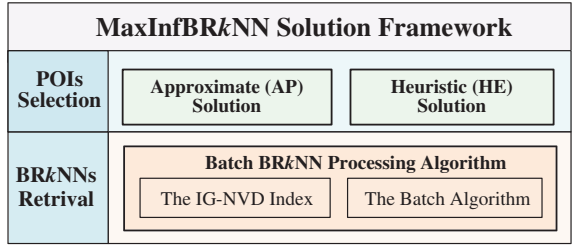}
\vspace{-3mm}
\caption{Framework Overview}
\label{fig:overview}
\vspace{-3mm}
\end{figure}

\begin{figure*}[t]
\centering
\includegraphics[width=6.8in]{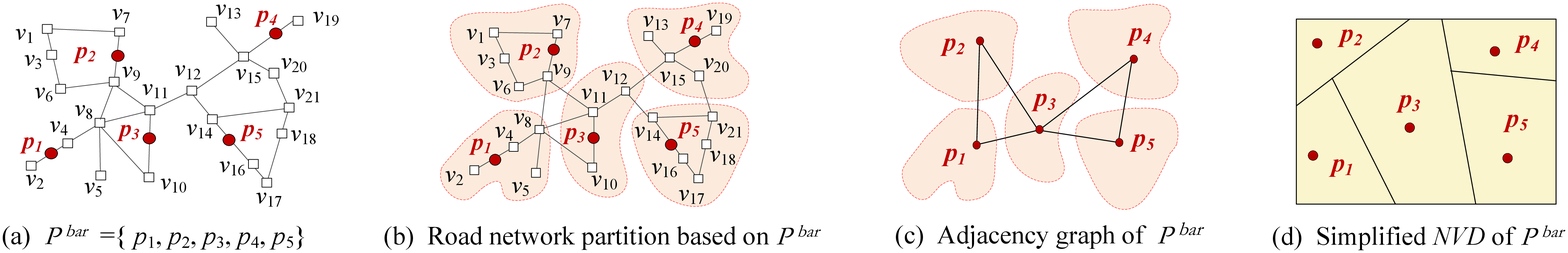}
\vspace*{-3mm}
\caption{{Illustration of NVD structures built for $\mathcal{P}$}}
\label{fig:nvd-IL-nvd-example}
\vspace*{-4mm}
\end{figure*}

\begin{figure}[t]
\centering
\hspace{-2mm}
\includegraphics[width=2.8in]{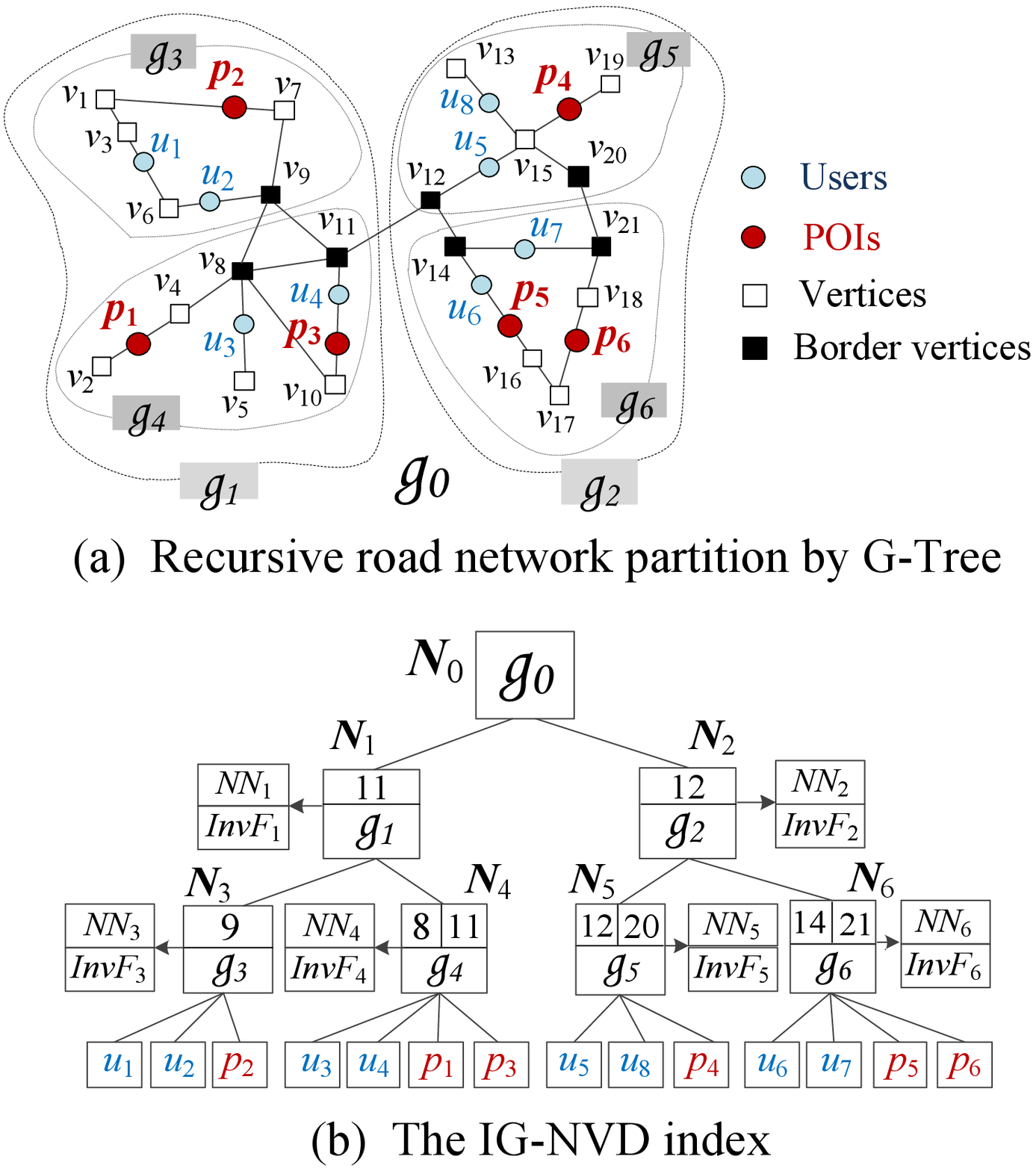}
\vspace{-3.5mm}
\caption{Illustration of the IG-NVD}
\label{fig:IG-NVD}
\vspace{-4.5mm}
\end{figure}

\textbf{Discussion.} Although the baseline solution can solve our problem with theoretical guarantees, it is still infeasible in practice due to three reasons. i) The time complexity of existing BR$k$NNs retrieval in geo-social networks is $O(\mathcal{|P|}\cdot\mathcal{|U|}^2)$~\cite{zhao2017icde}, which fails to scale to large geo-social networks. ii) The index (i.e., the GIM-Tree) used in the baseline is a memory consuming index. This index has to maintain two large matrices \cite{zhao2018icde,zhao2017icde} for social relevance computation, whose space cost of is $O(|\mathcal{P}|\cdot |\mathcal{U}|)$, which might yield prohibitive indexing costs. iii) Directly using RIS to evaluate the influence of POIs would lead to unstable running performance, especially when $b$ is small in real-life settings (as indicated by our empirical results). To ensure the accuracy, more RR sets have to be generated by BA to qualify the results, which results in expensive sampling costs and lower query efficiency.

%% file: overview.tex
\section{The Overview of Our Framework}
\label{sec:framework_overview}

To overcome the deficiencies in BA, we propose a framework as shown in Fig. \ref{fig:overview}. This framework composes two levels of technical contributions. Specifically, the first level in our framework is to support efficient and scalable batch BR$k$NN retrieval (to be introduced in Section~\ref{sec:batch}). Retrieving BR$k$NNs for multiple POIs is a prerequisite in solving MaxInfBR$k$NN with approximation guarantees, but is still very costly in practice, hence, the techniques in this level ensures the scalability of all our proposals. The second level offers strategies to effectively evaluate the influence of users in BR$k$NNs and selects influential POIs. Two alternative solutions are presented, including the approximation solution (denoted as AP) that returns theoretically guaranteed results (to be elaborated in Section~\ref{subsec:AP}), and the heuristic solution (denoted as HE) that solves the problem more efficiently without theoretical guarantees but still ensures highly qualified results (to be introduced in Section~\ref{sec:heuristic}).

%% file: batch.tex
\section{Batch BR$k$NN Processing}
\label{sec:batch}
We proceed to introduce a new index scheme with accompanying pruning techniques. Based on these, we develop efficient and scalable algorithms for batch BR$k$NN processing.

\subsection{Indexing Scheme}

Given that the social information in a geo-social network is dynamic and incurs extra update costs if indexed, we index only the geo-textual information that is relatively stable in real scenarios. Our index is, called IG-NVD, which is a variant of the \underline{i}nverted \underline{G}-Tree~\cite{zhong2015tkde} that incorporates \underline{N}etwork \underline{V}oronoi \underline{D}iagram (NVD) structures~\cite{mohammad2004vldb}.

As shown in Fig.~\ref{fig:IG-NVD}, the backbone of IG-NVD is a G-Tree~\cite{zhong2015tkde}, which recursively partitions a road network into sub-networks. The root node of a G-Tree represents the whole road network $\mathcal{G}_{r}$, and each sub-network $\mathcal{G}_{r}^{i} =(\mathcal{V}_{r}^{i}, \mathcal{E}_{r}^{i}, \mathcal{B}_{i}, \mathcal{P}_i,\mathcal{U}_i)$ is represented by a tree node $N_i$. Here, $\mathcal{V}_r^i$, $ \mathcal{E}_r^i$, and $\mathcal{B}_i$ denote the sets of vertices, edges, and border vertices in  $\mathcal{G}_r^i$, and  $\mathcal{P}_i$ ($\mathcal{U}_i$) is the set of POIs (users) located in $\mathcal{G}_r^i$.  Each tree node $N_i$ is associated with an inverted file and several NVDs. The inverted file of each node (e.g., \textit{Inv}$F_1$ of $N_1$ in Fig.~\ref{fig:IG-NVD}(b)) describes the textual information of users and POIs in the child nodes, which facilitates keyword-aware road network access.

The NVDs are constructed by dividing the road network according to the locations of POIs in the road network $\mathcal{G}_{r}$. Specifically, for a set $\mathcal{P}^{t}$ of POIs located on $\mathcal{G}_{r}$ covering the keyword $t$, the NVD for $\mathcal{P}^{t}$ is a data structure divides $\mathcal{G}_{r}$ into disjoint partitions, such that each partition corresponds to a specific POI $p\in\mathcal{P}^{t}$ and the set of vertices having $p$ as their nearest neighbor (NN). Let $\textit{NVD}^{t}$ denote the NVD built for $\mathcal{P}^{t}$. $\textit{NVD}^{t}$ can be modeled as an adjacency graph, whose node $\textit{NVP}(p_i)$ corresponds to a network Voronoi partition (NVP) for each $p\in\mathcal{P}^{t}$, and the adjacent edge connects two nodes $\textit{NVP}(o_i)$ and $\textit{NVP}(o_j)$ if an edge $(v_i,v_j)\in \mathcal{E}_r$ exists that connects $v_i\in \textit{NVP}(p_i)$ and $v_j\in \textit{NVP}(p_j)$. Figs.~\ref{fig:nvd-IL-nvd-example}(a) to \ref{fig:nvd-IL-nvd-example}(d) show an NVD example built for $\mathcal{P}^{\textit{bar}}=\{p_1,p_2,p_3,p_4,p_5\}$. As shown in Figs.~\ref{fig:nvd-IL-nvd-example}(b) and~\ref{fig:nvd-IL-nvd-example}(c), $\mathcal{G}_{r}$ is divided into five regions (NVPs), The NVPs of $\{p_1,p_2,p_3\}$ are mutually adjacent in the NVD. Each NVP contains the vertices with the same NN. By using NVDs, the NN of any vertex $v\in \mathcal{G}_r$ can be obtained directly. For example, the NN of a user located on $v_1$ is $p_2$, since $v_1 \in \textit{NVP}(p_2)$.

Let \textit{Voc} be the vocabulary of keywords in $\mathcal{P}$. For each $t_i\in\textit{Voc}$, we build NVDs for each $\mathcal{P}^{t_i}$ and associate the NN information with each vertex in $\mathcal{G}_{r}$. By using the properties of NVDs~\cite{mohammad2004vldb}, $k$ closest relevant POIs can be obtained to derive tightened score bounds to facilitate pruning (to be introduced in the Section~\ref{subsec:pruning}). With these  components, a nonleaf node $N_i$ of IG-NVD, is of the form ($\mathcal{G}_r^i$, $Ptr$, $\textit{InvF}_i$, $\mathcal{B}_i$, $\textit{NN}_i$). Here, (i) $\mathcal{G}_r^i$ is the subgraph indexed by $N_i$. (ii) $Ptr$ is a set of pointers pointing to the child nodes of $N_i$. (iii) $\textit{InvF}_i$ is the inverted file of $N_i$ indexing keyword information of users and POIs in child nodes, (iv) $\mathcal{B}_i$ is the set of border vertices of $N_i$, {and (v) $\textit{NN}_i$ is a map maintained for $N_i$ associating each $b\in\mathcal{B}_i$ with the corresponding NN in the NVDs. For example, in Figs~\ref{fig:IG-NVD}(a) and \ref{fig:IG-NVD}(b), given a tree node $N_1$ with $\mathcal{B}_1=\{v_{11}\}$, the NN of $v_{11}$ in $\textit{NVD}^{\textit{bar}}$ is $p_3$; so the record $\textit{NN}[v_{11}][\textit{`bar'}]=p_3$ is maintained in $\textit{NN}_1$. A leaf node $N_j$ is in a form ($\mathcal{G}_r^j$, $\textit{InvF}_j$, $\mathcal{B}_j$, $\mathcal{V}_r^j$, $\textit{NN}_j$), with $\mathcal{V}_r^j$ being the set of vertices in $\mathcal{G}_r^j$, $\textit{NN}_j$ associating each $v\in\mathcal{V}_r^i$ with NN information, and the remaining of elements being similar to those in  nonleaf nodes.}

\textbf{Optimizations.} Indexing all NVDs and NN information without excessive space costs is non-trivial, especially when the road network and the keyword vocabulary are large. To reduce the space cost, we introduce two optimizations. i) To reduce the number of required NVDs, we only build NVDs for frequent keywords and keep the information of infrequent keyword in the inverted posting lists~\cite{tenindra2020tkde}.
ii) As most of the frequent keywords are distributed unevenly in the road network, many vertices in $\mathcal{G}_r$ have the identical NN results. For such keywords, we only index compressed NN information by sharing their redundant records, thus reduce the average space cost required by each single NVD. Finally, instead of using the G-Tree to compute shortest path distances (whose time complexity is $O(|\mathcal{V}_r|)$~\cite{zhong2015tkde}), we utilize the pruned highway labeling (PHL) technique~\cite{akiba2014alenex} for network distance computation, whose response time is much smaller and nearly constant on massive road networks.

 \vspace{-1mm}
\subsection{Pruning Schemes}
\label{subsec:pruning}

Existing techniques~\cite{choudhury2016vldb,lu2011sigmod} derive separate bounds on spatial, textual, and social similarity scores~\cite{zhao2017icde}, and then combine them to obtain overall estimated bounds. This simple bound aggregation may yield loose overall bounds, which cause poor pruning abilities, especially for sparse social data. To overcome this, according to the definitions of BR$k$NN, we group relevant users of each $p\in\mathcal{P}_c$ into two categories: the set $\mathcal{U}_{S}(p)$ of users with social relevance to $p$ and the set $\mathcal{U}_{T}(p)$ of users with textual relevance. 
Users in different groups are evaluated separately with tailored pruning rules.
\vspace{1mm}

\subsubsection{Pruning Socially Relevant Users}

\begin{observation}
\label{obser:social_spacity}
In real-life settings, check-in data is sparse due to the limited mobility of users (i.e., users often visit nearby places). Thus, most POIs would have only a limited number of check-ins, and the average number of socially relevant users for each POI is relatively small w.r.t. $|\mathcal{U}|$.
\end{observation}

The sparsity of social data may impede bound estimation. Thus, instead of combining social relevance bounds with other score bounds, we extract socially relevant users, and bound their top-$k$ result scores separately with only geo-textual information. Due to Observation~\ref{obser:social_spacity}, this process can be efficient in practice (to be verified in Section~\ref{sec:exp}).

\vspace{1mm}
\begin{lemma}
\label{prune:user}
Given a user $u_i$, the scoring lower bound list $\mathcal{L}^{\downarrow}_{\textit{SB}}(u_i)$ of $u_i$ is a list of $k$ tuples $(p_j$, $s_j)$  $(1 \le j \le k)$ sorted in descending order of $s_j$. Here, $s_j$ is the geo-textual similarity score between $p_j$ and $u_i$, i.e.,  $s_j=(1-\alpha)\cdot \frac{f_t (u_i,p_j)}{ f_g(u_i,p_j)}$. For a POI $p$, if $F_{\textit{GST}}(u_i,p)  \textless  s_k$, 
we have $u_i \notin \mathcal{S}_r(p)$.
\end{lemma}

\begin{proof}
Assume that $u_i \in \mathcal{S}_r(p)$ (i.e., $u_i$ is in the BR$k$NN result set of $p$). According to Definition~\ref{dfn:rkgsk}, $p$ should be in the top-$k$ results of $u_i$, and thus we have $s_k \le F_{\textit{GST}}(u_i, p)$, which contradicts the condition of the lemma. \end{proof}

Algorithm~\ref{algo:usr_lbsl} shows the procedures of computing $\mathcal{L}^{\downarrow}_{\textit{SB}}(u)$ with the IG-NVD index. The algorithm  initializes $\mathcal{L}^{\downarrow}_{\textit{SB}}(u)$ and $\mathcal{PQ}$ as empty (line 1), where $\mathcal{PQ}$ is a priority queue to sort the keywords $t \in u.\textit{key}$ in descending order of their geo-textual scores. For each $t\in u.key$, a separate max-priority $\mathcal{H}_t$ is maintained to store relevant POIs $p$ (lines 2--9). Specifically, if $t$ is a frequent keyword, the nearest POI $p_\textit{NN}$ of $u$ in $\mathcal{P}^{t}$ is first assessed and inserted into $\mathcal{H}_t$ (lines 4--7). If $t$ is infrequent, all POIs in the inverted list of $t$ are assessed (line 9). When $|\mathcal{L}^{\downarrow}_{\textit{SB}}(u)|\leq$ $k$, a while loop continues to assess more local relevant POIs by moving the top POI $p_n$ from $\mathcal{PQ}$ to $\mathcal{L}^{\downarrow}_{\textit{SB}}(u)$ (lines 12--15), and inserting each adjacent POI of $p_n$ in $\textit{NVD}^{t_m}$ into $\mathcal{H}_{t_m}$ for further evaluation (lines 16--19). Finally, $\mathcal{L}^{\downarrow}_{\textit{SB}}(u)$ is returned when it contains $k$ POIs (line 20). Next, we give a running example to illustrate how Algorithm~\ref{algo:usr_lbsl} works.

\begin{example}
\label{examp:sb_list}
As shown in Fig.~\ref{fig:LBSN-dataset-example}, the socially relevant users of $p_1$ are $\{u_2, u_8\}$. Given $k$ = 2, we first compute $\mathcal{L}^{\downarrow}_{\textit{SB}}(u_2)$. Since $u_2.key$ = $\{'bar'\}$, which is a frequent keyword, the nearest POI $p_{\textit{NN}}$ of $u_2$ in $\textit{NVD}^{\textit{bar}}$ is $p_2$ (shown in Fig~\ref{fig:nvd-IL-nvd-example}(b)), and $( p_2, 0.41 )$ is inserted into $\mathcal{H}_{\textit{bar}}$ (lines 4--7) with $\langle (p_2, 0.41), bar\rangle$ added to $\mathcal{PQ}$ (line 11). Next, a while-loop is conducted to complete $\mathcal{L}^{\downarrow}_{\textit{SB}}(u_2)$. In the first iteration, $\langle (p_2, 0.41), bar\rangle$ is popped out from $\mathcal{PQ}$. As the current best keyword is `bar', the top element of $\mathcal{H}_{\textit{bar}}$ , $( p_2, 0.41)$, is removed from $\mathcal{H}_{\textit{bar}}$ and is inserted into $\mathcal{L}^{\downarrow}_{\textit{SB}}(u_2)$ (lines 12--15), and we have  $\mathcal{L}^{\downarrow}_{\textit{SB}}(u_2)=\{( p_2,$ $0.41)\}$. Thereafter, the adjacent POIs of $p_2$ in $\textit{NVD}^{\textit{bar}}$ (i.e., $p_1$ and $p_3$ in Fig~\ref{fig:nvd-IL-nvd-example}(c)) are evaluated and inserted into $\mathcal{H}_{\textit{bar}}$ (i.e., $\mathcal{H}_{\textit{bar}}=\{( p_1, 0.26 ), ( p_3, 0.17)\}$), and $\langle (p_1, 0.26), bar \rangle$ is added to $\mathcal{PQ}$ (lines 16--19). In the second iteration, $(p_1, 0.26)$ is removed from $\mathcal{H}_{bar}$, and is inserted into $\mathcal{L}^{\downarrow}_{\textit{SB}}(u_2)$. As $|\mathcal{L}^{\downarrow}_{\textit{SB}}(u_2)|=2$, the while-loop stops, with $\mathcal{L}^{\downarrow}_{\textit{SB}}(u_2)=\{( p_2,$ $0.41),$ $( p_1, 0.26)\}$. Similarly, we obtain $\mathcal{L}^{\downarrow}_{\textit{SB}}(u_8)=\{(p_4, 0.33),$ $(p_5,$ $0.18) \}$. According to Lemma~\ref{prune:user},  $u_8$ can be pruned for $p_1$, as  $F_{\textit{GST}}(u_8,$ $p_1)=0.08<\mathcal{L}^{\downarrow}_{\textit{SB}}(u_8).s_2=0.18$.
\end{example}

\begin{algorithm}[t]
\caption{$\mathcal{L}^{\downarrow}_{\textit{SB}}(u)$  computation algorithm}
\label{algo:usr_lbsl}
\begin{small}
\KwIn{a user $u$, the IG-NVD index, an integer $k$}
\KwOut{$\mathcal{L}^{\downarrow}_{\textit{SB}}(u)$ }
\vspace{0.5mm}
initialize $\mathcal{L}^{\downarrow}_{\textit{SB}}(u)$ and $\mathcal{PQ}$ as empty\;
\For {each keyword $t \in u.key$}{
        initialize $\mathcal{H}_t$ as empty\;
        \If{$t$ is a frequent keyword}{
             $p_{\textit{NN}}\gets$ the POI $p$ in $\textit{NVD}^{t}$ with $u.loc\in \textit{NVP}^{t}(p)$\;
             $s_{\textit{NN}}\gets (1-\alpha)\cdot f_t (u,p_{\textit{NN}}) \cdot f_g(u,p_{\textit{NN}})^{-1}$\;
             insert $( p_{\textit{NN}}, s_{\textit{NN}} )$ into $\mathcal{H}_t$\;
        }
         \Else{
            \hspace{2mm}insert all $p_{i}\in inv\_list(t)$ into $\mathcal{H}_{t}$\;

        }
        push the top element of $\mathcal{H}_t$ into $\mathcal{PQ}$;

}

\While{$|\mathcal{L}^{\downarrow}_{\textit{SB}}(u)| < k$}{
    $\left\langle (p_n,s_{n}),t_m \right \rangle \gets \mathcal{PQ}.dequeue()$\;
    \If{$p_n \notin \mathcal{L}^{\downarrow}_{\textit{SB}}(u)$ }{
       insert $( p_{n}, s_{n} )$ into $\mathcal{L}^{\downarrow}_{\textit{SB}}(u)$\;
    }
    remove $( p_n, s_n )$ from  $\mathcal{H}_{t_m}$\;
    \If{$t_m$ is a frequent keyword }{
        \For {each $p_x$ adjacent to  $p_n$ in $\textit{NVD}^{t_{m}}$}{
               insert $( p_x, s_{x} )$ into $\mathcal{H}_{t_m}$\;
        }
    }

    insert the top entry of $\mathcal{H}_{t_m}$ back into $\mathcal{PQ}$\;
}

\Return $\mathcal{L}^{\downarrow}_{\textit{SB}}(u)$\;
\end{small}
\vspace{-0mm}
\end{algorithm}

\vspace{1mm}
\subsubsection{Pruning Textually Relevant Users}
\label{subsub: prune_textual_user}

\indent

Let $\mathcal{U}_T(p)$ denote the set of users with textual relevance to a POI $p$. Intuitively, we can also utilize Algorithm~\ref{algo:usr_lbsl} to prune the users in $\mathcal{U}_T(p)$. Nonetheless, the number of users with textual relevance to $p$ could be considerably large, which means $|\mathcal{S}_r(p)| \ll |\mathcal{U}_{T}(p)|$. Therefore, if we evaluate each  $u\in\mathcal{U}_T(p)$ one by one, it would cause unnecessary computation costs. In view of this, we develop pruning schemes tailored for textually relevant users with the notion of pseudo-users.

\vspace{1mm}
\begin{definition}\label{defn:pseudo-user}
{\bf \textit{Pseudo-user}}. Given a POI $p$, a G-Tree node $N_i$, and a border vertex $b_i\in\mathcal{B}_i$, a pseudo-user $u_{ps}=(\textit{loc}, \textit{Area}$, $\textit{key})$ is a synthetic data point built on $b_i$, with  $u_{\textit{ps}}.\textit{loc}=b_i$, $u_{\textit{ps}}. \textit{Area}=N_i$, and $u_{\textit{ps}}.\textit{key}=\mathcal{U}_i.\textit{key}_u\cap p.\textit{key}$. Here $\mathcal{U}_i$ is the set of users in $N_i$, and $\mathcal{U}_i.key_u$ is the union of the keywords covered by all users $u\in\mathcal{U}_i$. 
\end{definition}

Unlike pruning Euclidean spaces~\cite{choudhury2016vldb}, the notion of pseudo-users is proposed to prune road network spaces (which is more complex) by verifying a few landmark border vertices along shortest paths. By utilizing pseudo-users, tighter bounds can be derived on road networks to facilitated the pruning lemmas presented as below.

\begin{algorithm}[t]
\caption{$\mathcal{M}^{\downarrow}_{\textit{SB}}(u_{ps})$ computation algorithm}
\label{algo:sb_map}
\small
\KwIn{a POI $p$, an integer $k$, a pseudo-user $u_{ps} = (b, N_i,$ $\mathcal{U}_i.key_u\cap p.key)$, the IG-NVD index}
\KwOut{$\mathcal{M}^{\downarrow}_{\textit{SB}}(u_{ps})$}
initialize $\mathcal{L}_{\textit{SB}}^{\downarrow}(t)$ as empty for each $t\in u_{\textit{ps}}.\mathit{key}$\;
\vspace{0.5mm}
\For {each  $t \in u_{\textit{ps}}.key$}{
     \If{$t$ is a frequent keyword}
    {
        $p_{\textit{NN}} \gets$ $p_i$ in $\textit{NVD}^{t}$ with $u_{ps}.loc\in \textit{NVP}^{t}(p_i)$\;
        $s_{\textit{NN}}\gets (1-\alpha)\times \textit{TS}(t,p_n) / dist(u_{ps},p_{\textit{NN}})$\;
        insert $( p_{\textit{NN}}, s_{\textit{NN}}) $ into $\mathcal{H}_t$\;
        \While{$|\mathcal{L}_{\textit{SB}}^{\downarrow}(t)|< k$}{
            $( p_n,s_n ) \gets \mathcal{H}_t.dequeue()$\;
            \If{$p_n$ not already evaluated}{
                insert $( p_n,s_{n} ) $ into $\mathcal{L}_{\textit{SB}}^{\downarrow}(t)$ \;
            }
            \For {each $p_{a}$ adjacent to $p_n$ in $\textit{NVD}^{t}$}{
              insert $( p_{a}, F_{\textit{GST}}(u_{ps},p_{a}) )$ into $\mathcal{H}_{t}$\;
            }

        }
    }
    \Else
    {
        insert top-$k$ POIs from $inv\_list(t)$ into $\mathcal{L}_{\textit{SB}}^{\downarrow}(t)$\;

    }
    maintain $\langle t, \mathcal{L}_{\textit{SB}}^{\downarrow}(t) \rangle$ into $\mathcal{M}^{\downarrow}_{\textit{SB}}(u_{\textit{ps}})$\;
}
\Return $\mathcal{M}^{\downarrow}_{\textit{SB}}(u_{\textit{ps}})$\;
\end{algorithm}

\begin{lemma}
\label{prune:upperbound}
Given a POI $p$, a pseudo-user $u_{\textit{ps}}=(b_i, N_i, \mathcal{U}_i \\.\textit{key}_u\cap p.\textit{key})$, where $b_i$ is one of border vertices of $N_i$ (i.e., $b_i\in \mathcal{B}_i$). Let $\mathcal{U}_i^{b}(p)$ be the users in $\mathcal{U}_i$, who are textually relevant to $p$ and whose shortest paths to $p$ pass through $b$. Then, for $\forall{u \in \mathcal{U}_i^{b}(p)}$, we have $F_{\textit{GST}}(u,p)\leq F_{\textit{GST}}(u_{ps},p)$.
\end{lemma}

\begin{proof}
The shortest path from any $u\in\mathcal{U}_i^{b}$ to $p$ must pass through $b$; thus,   $f_g(u_{ps},p)\leq f_g(u,p)$ for each $u\in \mathcal{U}_i^{b}(p)$. Next, as $\mathcal{U}_i^{b}(p).\textit{key}_u \subset \mathcal{U}_i.\textit{key}_u$, $f_t(u_{\textit{ps}},p)\geq f_t(u,p)$ holds for each $u\in \mathcal{U}_i^{b}(p)$. Combined  with Eq.~\ref{equ:gsk}, for each $u\in \mathcal{U}_i^{b}(p)$, we have $F_{\textit{GST}}(u,p)$ $\leq  F_{\textit{GST}}(u_{ps},p)$. \end{proof}

Given a pair of a POI $p$ and a pseudo-user $u_{ps}=(b,N_i,$ $\mathcal{U}_i.key_u\cap p.key)$, we build a hash map $\mathcal{M}^{\downarrow}_{\textit{SB}}(u_{ps})$ associating each keyword $t\in u_{ps}.key$ with a sorted list $\mathcal{L}^{\downarrow}_{\textit{SB}}(t)$ that contains $k$ tuples $\langle p_i$, $s_i\rangle$ ($1 \le i \le k$). Here, $p_i$ is a POI with keyword $t$, and $s_i = (1-\alpha)\cdot TS(t,p_i) / dist(b,p_i)$, where $TS(t,p_i)$ is the TF-IDF score of $t$ in $p_i$.

\begin{lemma}
\label{lemma:prune_keyword}
Given a pseudo-user $u_{\textit{ps}}=(b,N_i,$ $\mathcal{U}_i.\mathit{key}_u\cap p.\mathit{key})$, if $\mathcal{L}^{\downarrow}_{\textit{SB}}(t).s_k\geq F_{\textit{GST}}(u_{ps},p)$ holds for each ${\langle t,\mathcal{L}^{\downarrow}_{\textit{SB}}(t)\rangle} \in \mathcal{M}^{\downarrow}_{\textit{SB}}(u_{ps})$, none of textually relevant users in $\mathcal{U}_i^{b}(p)$ can be influenced by $p$, and thus $u_{ps}$ can be discarded.
\end{lemma}

\begin{proof}
By Lemma~\ref{prune:upperbound}, $F_{\textit{GST}}(u_{ps},p) \ge F_{\textit{GST}}(u,p)$ holds for any  $u\in\mathcal{U}_i^{b}(p)$. If all $t\in u_{ps}.key$ having $\mathcal{L}^{\downarrow}_{\textit{SB}}(t).s_k\geq F_{\textit{GST}}(u_{ps},p)$, it must hold that  $\mathcal{L}^{\downarrow}_{\textit{SB}}(t).s_k\geq F_{\textit{GST}}(u,p)$ for all users in $\mathcal{U}_i^{b}(p)$. Hence, there are at least $k$ POIs $p_i (1 \le i \le k)$ in $\mathcal{L}^{\downarrow}_{\textit{SB}}(t)$ with $F_{\textit{GST}}(u, p_i) \ge s_k \ge F_{\textit{GST}}(u_{ps},p) \ge F_{\textit{GST}}(u,p)$. Therefore, no user $u\in\mathcal{U}_i^{b}(p)$ can be influenced by $p$, and thus $u_{ps}$ can be pruned away. The proof completes. \end{proof}

\begin{lemma}
\label{lemma:prune-early-termination-beyond}
Given a tree node $N_i$, if all pseudo-users $u_{psi}$ generated for $p$ on each border $b_i\in N_i$ satisfy the condition in Lemma~\ref{lemma:prune_keyword}, all the textually relevant users with shortest paths to $p$ passing any vertex in ${N}_i$ cannot be influenced by $p$.
\end{lemma}

\begin{proof}
Lemma~\ref{lemma:prune-early-termination-beyond} is an extension of Lemma~\ref{lemma:prune_keyword}. \end{proof}

\vspace{-1mm}
The pseudo-code of Algorithm~\ref{algo:sb_map} summarizes the procedure of computing  $\mathcal{L}_{\textit{\textit{SB}}}^{\downarrow}(t)$ for each frequent keyword $t\in u_{\textit{ps}}.\textit{key}$ when generate  $\mathcal{M}^{\downarrow}_{\textit{SB}}(u_{\textit{ps}})$, the workflow resembles that in Algorithm~\ref{algo:usr_lbsl}, so we omit its details for brevity.

\subsection{Batch Processing Algorithms}

Exploiting the above pruning rules, the pseudo-code of the batch BR$k$NN processing algorithm is presented in Algorithm~\ref{algo:batch}. Two important functions of Algorithm~\ref{algo:batch} are shown in Algorithm~\ref{algo:core_function}. Algorithm~\ref{algo:batch} initializes a priority queue $\mathcal{Q}$ to maintain pairs of each generated pseudo-user and the corresponding POI in ascending order of the shortest path distances. A set $\mathcal{U}_c$ stores all BR$k$NN result candidates, and a vector $\mathbb{N}_{\textit{max}}$ maintains the current uppermost nodes explored for each POIs. A map $\mathcal{AL}$ associates each $u\in\mathcal{U}_c$ with the POIs in $\mathcal{P}_c$ that may influence $u$ (line 1). Let $\mathcal{U}_{T}(p, \textit{Leaf}(p))$ denote the set of local textually relevant users of a POI $p$ (i.e., the users located in the same leaf node of the POI). Algorithm~\ref{algo:batch} first evaluates socially relevant users and users in $\mathcal{U}_{T}(p, \textit{Leaf}(p))$ for each $p\in\mathcal{P}_c$ (lines 2--7). Then, it iteratively expands the road network and evaluate remaining textually relevant users to find additional BR$k$NN candidates (lines 8--13). During the process, pseudo-users are generated on the border vertices in different levels of the G-Tree (lines 6--10 of Algorithm~\ref{algo:core_function}) to effectively prune users by Lemmas~\ref{lemma:prune_keyword} and~\ref{lemma:prune-early-termination-beyond} (lines 12 and 17 of Algorithm~\ref{algo:core_function}). Users who cannot be pruned are added to $\mathcal{U}_c$ (line 6) and further verified by the top-$k$ algorithm extended from~\cite{tenindra2020tkde} (lines 15). Note that we modify the existing top-$k$ algorithm \cite{tenindra2020tkde} to incorporate social relevance. Finally, batch results for each $p\in\mathcal{P}_c$ are returned (line 18). To better illustrate the procedures in Algorithm~\ref{algo:batch}, we give a running example as below.


\begin{example}
\label{examp:bath}
Given $\mathcal{P}_c=\{p_1,p_3\}$ and $k=2$.
Algorithm~\ref{algo:batch} first evaluates socially and local textual relevant users (lines 2--7). Since each step is similar to Example~\ref{examp:sb_list}, we omit its details for brevity. After that, we obtain $\mathcal{U}_c=\{u_2, u_3, u_4, u_5\}$, $\mathbb{N}_{\textit{max}}=\{\langle p_1,N_4 \rangle,\langle p_3,N_4 \rangle\}$, and $\mathcal{AL}=\{(u_2, \langle p_1\rangle), (u_3, \langle p_1, p_3\rangle),$ $(u_4, \langle p_3 \rangle), (u_5, \langle p_3 \rangle)\}$. The rest of textually relevant users are evaluated as follows.

\vspace{1mm}
\textbf{Loop 1:} As $\mathcal{Q}$ is empty, $\mathbb{N}_{\textit{max}}=\{\langle p_1,N_4 \rangle, \langle p_3,N_4 \rangle \}$, Extend\_Range function is invoked to explore the G-tree (line 9). In Fig~\ref{fig:IG-NVD}, the  nearest ancestor node of $N_4$ with textual relevance to $p_1$ and $p_3$ is $N_1$ (line 4--5 of Algorithm~\ref{algo:core_function}). And the child nodes of $N_1$, (which is also the sibling node of $N_4$) with textual relevance to $p_1$ and $p_3$ is $N_3$. Hence, two pseudo-users $u_{ps1}=(v_{9}, N_3, \{\textit{`bar'}, \textit{`wine'}\})$ and $u_{ps2}=(v_{9}, N_3, \{\textit{`bar'}\})$ are generated on the border vertex $v_9$ for $p_1$ and $p_3$, respectively, with $(u_{ps1},p_1)$ and $(u_{ps2},p_3)$ inserted into $\mathcal{Q}$ (i.e., $\mathcal{Q}=\{ \langle (u_{ps2},p_3), 5\rangle, \langle (u_{ps1},p_1), 10\rangle\}$) (line 6--10 of Algorithm~\ref{algo:core_function}). Now, $N_1$ becomes the current uppermost nodes explored for $p_1$ and $p_3$. Since $N_1$ can be pruned for $p_1$ but cannot be pruned for $p_3$ via Lemma~\ref{lemma:prune-early-termination-beyond}, we update $\mathbb{N}_{\textit{max}}=\{\langle p_3,N_1 \rangle\}$ (lines 11--13 of Algorithm~\ref{algo:core_function}).

Next, it proceeds to pop out $(u_{ps2},p_3)$ from $\mathcal{Q}$ (line 10), where $u_{ps2}=(v_{9}, N_3,$ $\{\textit{`bar'}\})$. $\mathcal{M}^{\downarrow}_{\textit{SB}}(u_{ps2})$ is generated by computing $\mathcal{L}^{\downarrow}_{\textit{SB}}('bar')$(lines 12). As the keyword of $u_{ps2}$, $'bar'$, can not be pruned by lemma~\ref{lemma:prune_keyword} and $N_3$ is a leaf node (lines 15--18 of Algorithm~\ref{algo:core_function}), the textually relevant users $u_1$ and $u_2$ in $N_3$ are evaluated for $p_3$ (line 21 of Algorithm~\ref{algo:core_function}). As $\mathcal{L}^{\downarrow}_{\textit{SB}}(u_{1}).s_2$ $\geq F_{\textit{GST}}(u_1,p_3)$ and  $\mathcal{L}^{\downarrow}_{\textit{SB}}(u_{2}).s_2$ $\geq F_{\textit{GST}}(u_2,p_3)$, both $u_1$ and $u_2$ can be pruned for $p_3$ by Lemma~\ref{prune:user} (line 6).

\textbf{Loop 2:} As $\mathcal{Q}$ is not empty, we proceed to pop out $(u_{ps1},p_1)$ from $\mathcal{Q}$, where $u_{ps1}=(v_{9}, N_3, \{\textit{`bar',`wine'}\})$. Similarly, $u_{ps2}$ can not be pruned, and thus textually relevant users in $N_3$ (i.e., $u_1$ and $u_2$) are evaluated for $p_1$. As $\mathcal{L}^{\downarrow}_{\textit{SB}}(u_{1}).s_2\leq F_{\textit{GST}}(u_1,p_3)$ and $\mathcal{L}^{\downarrow}_{\textit{SB}}(u_{2}).s_2\leq F_{\textit{GST}}(u_2,p_3)$, $u_1$ and $u_2$ are added to $\mathcal{U}_c$ for further verification.

Again, $\mathcal{Q}$ becomes empty, since $\mathbb{N}_{\textit{max}}\neq\varnothing$, we continue to explore the road network, update $\mathcal{Q}$, and pop out the remaining entries until $\mathcal{Q}=\varnothing$ and $\mathbb{N}_{\textit{max}}=\varnothing$. When the while loop stops, we get $\mathcal{U}_c=\{u_1,u_2,u_3,u_4,u_5\}$, which is further examined by the top-$k$ algorithm. Details of top-$k$ algorithm can be found in~\cite{tenindra2020tkde}. Finally, we obtain batch BR$k$NN results as  $S_r(p_1)=\{u_1,u_2,u_3\}$ and $ S_r(p_3)=\{u_3,u_4,u_5\}$.

\end{example}

\begin{algorithm}[t]
\caption{Batch processing algorithm}
\label{algo:batch}
\small
\KwIn{a query set $\mathcal{P}_c$, the IG-NVD index, an integer $k$}
\KwOut{batch BR$k$NN result set for each $p\in\mathcal{P}_c$}
initialize priority queue $\mathcal{Q}$, user set $\mathcal{U}_c$, a tree node vector $\mathbb{N}_{\textit{max}}$, and hash map $\mathcal{AL}$\;
\For{each $p \in\mathcal{P}_c$ }{
    \For {each user $u \in \mathcal{U}_{S}(p) \cup \mathcal{U}_{T}(p$, $\textit{Leaf}(p))$}{
         compute $\mathcal{L}^{\downarrow}_{\textit{SB}}(u)$ by using Algorithm ~\ref{algo:usr_lbsl}\;
         \If{$\mathcal{L}^{\downarrow}_{\textit{SB}}(u).s_{k}<F_{\textit{GST}}(u,p)$}{
                insert $u$ into $\mathcal{U}_c$, add $p$ to $\mathcal{AL}\left[u\right]$; // Lemma 4
         }
   
    }
    add $\langle p,$ $\textit{Leaf}(p)\rangle$ to $\mathbb{N}_{\textit{max}}$\;
}

\While{$\mathcal{Q} \neq \varnothing$ \textbf{or} $\mathbb{N}_{\textit{max}} \neq \varnothing$}{
    \textbf{if} \hspace*{0.01in} $\mathcal{Q}=\varnothing$ \hspace*{0.02in} \textbf{then} \hspace*{0.02in} Extend\_Range($\mathbb{N}_{\textit{max}}$,$\mathcal{Q}$, $k$)\;

    $\langle u_i, p_j \rangle \gets$ DEQUEUE$(\mathcal{Q})$\;
    \If {$u_i$ is not processed for $p_j$}{
       compute $\mathcal{M}^{\downarrow}_{\textit{SB}}(u_i)$\;
       UpdateQueue($\mathcal{M}^{\downarrow}_{\textit{SB}}(u_i)$, $p_j$, $\mathcal{Q}$, $k$)\;
    }
}
\For{each user $u\in\mathcal{U}_c$}{

   $\mathcal{S}_t(u)\gets$ T$k$GSKQ($u$, $k$); // algorithm of~\cite{tenindra2020tkde}\;
  
   \For{each $p_i \in \mathcal{AL}[u]$}{
        
         \textbf{if}  $p_i \in \mathcal{S}_t(u)$ \textbf{then} \hspace{2mm} add $u$ into $\mathcal{S}_r(p_i)$\;    
   }
}
\Return $\mathcal{S}_r(p)$ for each $p\in\mathcal{P}_c$\;
\end{algorithm}
\vspace{-1mm}

\vspace{0mm}
\subsection{Discussion.} 

We analyze the time complexity of Algorithm~\ref{algo:batch} which includes three parts: i) evaluating socially and nearby textually relevant users (lines 2--7), ii) evaluating remaining textually relevant users (lines 8--13), and iii) performing top-$k$ verification of users in BR$k$NN candidates (lines 14--17). In the worst case, Algorithm~\ref{algo:batch} at most accesses $|\mathcal{U}|$ users for each POI during part i) and ii), whose time complexity is $O(|\mathcal{U}||\mathcal{P}_c|\delta_{1})$, where $\delta_{1}$ denotes the time complexity of Algorithm~\ref{algo:usr_lbsl}. Let $\tau$ be the largest frequency of an infrequent keyword, and let $\bar{w}_1$ ($\bar{w}_2$) be the average number of frequent (infrequent) keywords of a user. By analyzing Algorithm~\ref{algo:usr_lbsl}, we have $\mathcal{O}(\delta_{1})$=$O(k(\bar{w}_1+\bar{w}_2)+(\tau +log\tau)\bar{w}_2+ klogk)$. Further, at most $|\mathcal{B}|$ (the border vertex set size) pseudo-users are generated for each POI, the complexity of which is $O(|\mathcal{B}||\mathcal{P}_c|\delta_{2})$. Here, $\delta_{2}$ represents the cost of computing  $\mathcal{M}^{\downarrow}_{\textit{SB}}(u_{ps})$ for a pseudo-user $u_{\textit{ps}}$, which is  $O(k\bar{W}_1+\tau\bar{W}_2 + klogk + \tau log\tau)$, with $\bar{W}_1$ ($\bar{W}_2$) being the average number of frequent (infrequent) keywords covered by a POI. Finally, let $\Delta$ be the cost of the top-$k$ algorithm. Then cost of phase iii) is $O(|\mathcal{U}|\mathcal{P}_c|\Delta)$. To sum up, the total time complexity of Algorithm~\ref{algo:batch} is $O(\xi_1|\mathcal{U}|+\xi_2)$, with $\xi_1=(\Delta+\delta_1)|\mathcal{P}_c|$ and $\xi_2=\delta_2|\mathcal{P}_c||\mathcal{B}|$. As $\tau$, $\bar{w}_1$, $\bar{w}_1$, $\bar{W}_1$, $\bar{W}_2$, $|\mathcal{P}_c|$ and $k$ are often small values, while $|\mathcal{B}|$ and $\Delta$ are also small in real settings~\cite{tenindra2020tkde,sibo2016vldb}, we have that, $\xi_1$ and $\xi_2$ are reasonably small as well (i.e., $\xi_1(\xi_2) \ll |\mathcal{P}|$).

\begin{algorithm}[t]
\caption{Key Functions in Algorithm~\ref{algo:batch}}
\label{algo:core_function}
\small
\hspace{-1.2mm}
\textbf{Function:} Extend\_Range$(\mathbb{N}_{\textit{max}}$, $\mathcal{Q}$, $k)$ \\

\For{each $\langle p_i, N_{i} \rangle \in \mathbb{N}_{\textit{max}}$}{
    $N_{k} \gets N_{i}.\textit{father}$\;
    \While{$\mathcal{U}_{k}.\textit{key}_{u} \cap p_i.\textit{key} = \varnothing$ and $N_{k} \neq root$ }{
        $N_{k} \gets N_{k}.\textit{father}$ \;
    }
     \For{each child $N_j$ of $N_k$ $(\mathcal{U}_j.key_u\cap u_{ps}.key =  \varnothing)$}
     {
         \If{$N_j$ is not Leaf$(p)$ and not its any ancestor}{
                            \For{each border vertex $b'\in N_j$}{
                                $u_{ps}' \gets (b',N_j, \mathcal{U}_j.key_u\cap p.key$\;
                                {\rm ENQUEUE}$(\mathcal{Q}$, $\langle u_{ps}', p\rangle$, $dist(b',p))$\;
                            }
        }
    }

    \If{$N_k$ can be discarded for $p_i$}{
        remove $\langle p_i, N_{i} \rangle$ from $\mathbb{N}_{\max}$; // {\rm Lemma~\ref{lemma:prune-early-termination-beyond}}
    }
    \textbf{else} \hspace{0.5mm} replace $\langle p_i, N_{i} \rangle$ with $\langle p_i, N_{k} \rangle$ into $\mathbb{N}_{\max}$\;

}

\hspace{-1.2mm}
  {\bf Function:} UpdateQueue($\mathcal{M}_{\textit{SB}}^{\downarrow}(u_{ps}=\langle b_i, N_i, key \rangle), p, \mathcal{Q}, k$)\\

 \For{each $ \langle t,  \mathcal{L}_{\textit{SB}}^{\downarrow}(t) \rangle \in \mathcal{M}_{\textit{SB}}^{\downarrow}(u_{ps})$}{

                  \If { $\mathcal{L}^{\downarrow}_{\textit{SB}}(t).s_{k} \geq      F_{GST}(u_{ps},p)$}{
                    
                        \hspace{-1mm} remove $t$ from $u_{ps}.key$; // {\rm Lemma 6}
                  }
}

\If{ $u_{ps}.key \neq \varnothing$}{
\If{$N_i$ is not leaf node}
{
execute lines 6--10 for $N_i$\;}
 \textbf{else} \hspace{0.5mm} execute lines 4--6 in Algorithm~\ref{algo:batch} for each $u\in\mathcal{U}_{T}(p,N_i)$\;
}
\vspace{0mm}
\end{algorithm}

%% file: poi.tex
\vspace{-2mm}
\section{POI Selection Policies}
\label{sec:poi}

Next, we present methods to effectively evaluate the influence of users in BR$k$NNs, and we give alternative policies to answer MaxInfBR$k$NN queries with qualified results.

\vspace{-1.5mm}
\subsection{Effective Influence Evaluation}
\label{sec:hybrid-estimator}

Given a POI set $\mathcal{P}_s$, and a RR set collection $\mathcal{R}$, instead of only using RR set coverage to derive the influence estimation $\hat{\mathbb{I}_p}(\mathcal{P}_{i})$, we evaluate the influence in a hybrid manner: 
\begin{equation}
\label{equa:hybrid-evaluate}
    \hat{\mathbb{I}_p}(\mathcal{P}_{s})=\mathbb{I}^{L}_{h}(\mathcal{P}_{s})+ \hat{\mathbb{I}}^{R}_{h}(\mathcal{P}_{s})=\mathbb{I}^{L}_{h}(\mathcal{P}_{s})+ \Lambda_{\mathcal{R}}\left(\mathcal{P}_{s}\right)\cdot|\mathcal{V}_s|/|\mathcal{R}|
\end{equation}

\noindent In Eq.~\ref{equa:hybrid-evaluate}, $\mathbb{I}^{L}_{h}(\mathcal{P}_{s})$ is the local influence of $\mathcal{P}_{s}$, which is propagated from users in $\mathcal{S}_r(\mathcal{P}_{s})$ at most $h$-hops. $\hat{\mathbb{I}}^{R}_{h}(\mathcal{P}_{i})$ is the remote influence estimation of $\mathcal{P}_{s}$, which approximates the number of users influenced by $\mathcal{S}_r(\mathcal{P}_{s})$ further than $h$-hops away, and is measured by the number (i.e., $\Lambda_{\mathcal{R}}(\mathcal{P}_{i})$) of covered RR sets. Using this formulation, though the exact social influence computation is hard to capture, the exact $h$-hop based local influence can be well and efficiently computed~\cite{tang2017asonam,tang2018snam}. As computing exact local influence is expensive when $h\geq3$, we fix $h=2$ as a trade-off between efficiency and accuracy~\cite{tang2017asonam,tang2018snam}. To estimate the remote influence with RR sets, it worth mentioning that the RR set should be modified to capture the influence beyond $2$-hops, so, any POI $p$ is removed from the original RR set when users in $\mathcal{S}_r(p)$ are reached during the sampling within 2 hops.

\vspace{-1mm}
\subsection{Approximation Solution}
\label{subsec:AP}

\textbf{Revisiting the bounds in Lemma~\ref{lemma:poi-inf-upper-lower-inf}.} We first greedily select $b$ POIs to maximize the objective function in Eq.~\ref{equa:hybrid-evaluate}, and denote this POI set as $\mathcal{P}_{b}^{H^{*}}$. We also extend an existing 2-hop based greedy algorithm~\cite{tang2017asonam} to generate a POI set with approximately optimal local influence, and we denote this POI set as $\mathcal{P}_{b}^{L^{*}}$. Let $\mathbb{P}_{b}$ be the collection of all size-$b$ subsets of $\mathcal{P}_c$, (i.e., $\mathbb{P}_{b}=\{\mathcal{P}_i|\mathcal{P}_{i} \subset \mathcal{P}_c \wedge |\mathcal{P}_{i}|=b\}$), and let  $\mathcal{P}_{b}^{H^{\textit{opt}}}$ ($\mathcal{P}_{b}^{L^{\textit{opt}}}$) be the POI set with the largest influence (local influence) among all elements in $\mathbb{P}_{b}$. Next, we revise the influence bounds as follows.

\vspace{1mm}
\textbf{i) Derivation of  $\mathbb{I}^{-}_{p}(\mathcal{P}_b^{H^{*}})$.} Combining Eq.~\ref{equa:inf_lower} of Lemma~\ref{lemma:poi-inf-upper-lower-inf} with Eq.~\ref{equa:hybrid-evaluate}, we get the following revised lower bound:
\vspace{-0.5mm}
\begin{small}
\begin{equation}
\begin{split}
\label{equa:hybrid-lowerbound-evaluate}
&\mathbb{I}^{-}_{p}\left(\mathcal{P}_{b}^{H^{*}}\right) = \mathbb{I}^{L}_{h}(\mathcal{P}_{b}^{H^{*}}) + \\
&\qquad  \left(\left(\sqrt{\Lambda_{\mathcal{R}_{2}}\left(\mathcal{P}_{b}^{H^{*}}\right)+\frac{2 \eta_{l}}{9}}-\sqrt{\frac{\eta_{l}}{2}}\right)^{2}-\frac{\eta_{l}}{18}\right) \cdot \frac{|\mathcal{V}_s|}{|\mathcal{R}_{2}|}
\end{split}
\end{equation}
\end{small}
\textbf{ii) Derivation of  $\mathbb{I}^{+}_{p}(\mathcal{P}_b^{H^{\textit{opt}}})$.}
As submodularity also holds for the local influence computation \cite{tang2018snam}, we have $\mathbb{I}^{L}_{h}(\mathcal{P}_{b}^{H^{\textit{opt}}})\leq\mathbb{I}^{L}_{h}(\mathcal{P}_{b}^{L^{\textit{opt}}})\leq\mathbb{I}^{L}_{h}(\mathcal{P}_{b}^{L^{*}})/(1-1/e)$. By combining this with Eqs.~\ref{equa:hybrid-evaluate} and~\ref{equa:inf_upper}, we derive the following  upper bound:
\vspace{-0.5mm}
\begin{small}
\begin{equation}
\begin{split}
\label{equa:hybrid-upperbound-evaluate}
&\mathbb{I}^{+}_p\left(\mathcal{P}_{b}^{H^{opt}}\right) = \mathbb{I}^{L}_{h}(\mathcal{P}_{b}^{L^{*}})/(1-1/e) + \\
& \qquad \left(\sqrt{\Lambda_{\mathcal{R}_{1}}^{u}\left(\mathcal{P}_{b}^{H^{opt}}\right)+\frac{\eta_{u}}{2}}+\sqrt{\frac{\eta_{u}}{2}}\right)^{2} \cdot \frac{|\mathcal{V}_s|}{|\mathcal{R}_{1}|}
\end{split}
\end{equation}
\end{small}

The maximum number of RR sets ensuring the approximation ratio is formulated in Lemma \ref{lemma:rrset-size-bound}, revised from \cite{tang2018sigmod}. 

\begin{lemma}[\cite{tang2018sigmod}]
\label{lemma:rrset-size-bound}
Given a set $\mathcal{R}$ of RR sets and parameters $\epsilon$ and $\delta$, if $|\mathcal{R}|$ is larger than $\theta_{\textit{max}}$:\\
\vspace{0.5mm}
\begin{small}
$ |\mathcal{R}| \geq \theta_{\textit{max}}=  \frac{2 |\mathcal{V}_s|\cdot \left((1-1 / e) \sqrt{\ln \frac{6}{\delta}}+\sqrt{(1-1/e)\left(\ln \left(\begin{array}{l} |\mathcal{P}_c| \\ b \end{array}\right)+\ln \frac{6}{\delta}\right)}\right)^{2}}{\varepsilon^{2} \mathbb{I}_p(\mathcal{P}_{b}^{H^{opt}})}$,
\end{small}

\noindent then the approximation ratio of $\mathcal{P}_{b}^{H^{*}}$ is no smaller than $1-1/e-\epsilon$ with at least $1-\delta/3$ probability.

\end{lemma}

To approximate the value of $\mathbb{I}_p(\mathcal{P}_{b}^{H^{opt}})$ in  Lemma~\ref{lemma:rrset-size-bound}, an efficient and tightened lower bound computation can be  $\mathbb{I}^{L}_{h}(\mathcal{P}_{b}^{L^{*}})$ (denoted as $\Psi$). According to $\theta_{\max}$, we can first sample $\theta_{0}$ random RR sets to initiate the RIS process, where $\theta_{0}$ is computed as follows.
\begin{small}
\begin{equation}
\label{equa:theta0}
 \hspace{15mm} \theta_{0} = \theta_{\textit{max}} \cdot \epsilon^{2} \cdot \Psi /|\mathcal{V}_s| \\
\end{equation}
\end{small}


\begin{algorithm} [t]
\label{alg:approx}
  \caption{Approximation (AP) algorithm}
  
\small
  \KwIn{the IG-NVD index, query parameters   $\mathcal{P}_c$, $k$, and $b$, and the system parameters $\delta$ and $\epsilon$}
  \KwOut{a size-$b$ POI set $\mathcal{P}_{s}$}
  
  obtain $S_r(\mathcal{P}_c)$ by using Algorithm~\ref{algo:batch};  // retrieve the batch BR$k$NN results for each $p\in\mathcal{P}_c$\;
 construct the heterogeneous graph $\mathcal{G}_{H}$\;
initialize $\mathcal{P}_{b}^{L^{*}}$ and $\mathcal{P}_{b}^{H^{*}}$ as empty\;
generate $\mathcal{P}_{b}^{L^{*}}$ using the 2-hop based greedy algorithm \cite{tang2017asonam}\;
initialize $\theta_{\textit{max}}$ and  $\theta_{0}$ based on Lemmas~\ref{lemma:rrset-size-bound} and Equation~\ref{equa:theta0}\;
generate two RR set collections $\mathcal{R}_1$ and $\mathcal{R}_2$ with size $\theta_{0}$\;

$i_{\max} \leftarrow\left\lceil\log _{2} \frac{\theta_{\max }}{\theta_{0}}\right\rceil$, set $\eta_{l}$ and $\eta_{u}$ in Eqs.~\ref{equa:hybrid-lowerbound-evaluate} and~\ref{equa:hybrid-upperbound-evaluate} as $\ln(3i_\textit{max}/\delta)$\;
\For{$i = 1$ \textbf{to} $i_{\max} $}{
    get $\mathcal{P}_{b}^{H^{*}}$ by greedy algorithm revised by  Eq.~\ref{equa:hybrid-evaluate} on $\mathcal{R}_1$\;

    compute lower bound  $\mathbb{I}^{-}_p\left(\mathcal{P}_{b}^{H^{*}}\right)$ using $\mathcal{R}_2$ by Eq.~\ref{equa:hybrid-lowerbound-evaluate}\;

    compute upper bound  $\mathbb{I}^{+}_p\left(\mathcal{P}_{b}^{H^{opt}}\right)$ using  $\mathcal{R}_1$ and $\mathcal{P}_{b}^{L^{*}}$ by Eq.~\ref{equa:hybrid-upperbound-evaluate}\;

    $ratio \gets \mathbb{I}^{-}_p\left(\mathcal{P}_{b}^{H^{*}}\right)/ \mathbb{I}^{+}_p\left(\mathcal{P}_{b}^{H^{opt}}\right)$\;
    \If{$ratio\geq 1- 1/e-\epsilon$  or  $i=i_{\textit{max}}$}{
        $\mathcal{P}_{s} \gets \mathcal{P}_{b}^{H^{*}}$ and break\;

    }
    double the sizes of $\mathcal{R}_1$ and $\mathcal{R}_2$ by merging new RR sets\;
}
\textbf{return} $\mathcal{P}_{s}$\;
\end{algorithm}


With above formulations, the approximation solution (\textbf{denoted as AP}) to MaxInfBR$k$NN is presented in Algorithm \ref{alg:approx}. AP first invokes Algorithm~\ref{algo:batch} to retrieve batch results for each $p\in\mathcal{P}_c$, and then constructs the  heterogeneous graph $\mathcal{G}_{H}$ (lines 1--2). Next, it obtains $\mathcal{P}_{b}^{L^{*}}$ to derive a tight lower bound on $\mathbb{I}_p(\mathcal{P}_{b}^{H^{opt}})$, and then it initializes $\theta_{\max}$ and $\theta_{0}$ by Lemma~\ref{lemma:rrset-size-bound} and Eq.~\ref{equa:theta0}, respectively (lines 3--5). After this, AP samples two RR set collections $\mathcal{R}_1$ and $\mathcal{R}_2$ with $|\mathcal{R}_1|=|\mathcal{R}_2|=\theta_0$ to initiate the sampling phase with at most $i_{\max}=\left\lceil\log _{2} \frac{\theta_{\max }}{\theta_{0}}\right\rceil$ iterations (lines 8--15). 
In each iteration, AP utilize $\mathcal{R}_{1}$ to obtain $\mathcal{P}^{H^{*}}_{b}$ which approximately maximizes the objective function in Eq.~\ref{equa:hybrid-evaluate} (line 9). Next, it utilizes $\mathcal{R}_2$ to derive $\mathbb{I}^{-}_p(\mathcal{P}^{H^{*}}_{b})$ by Eq.~\ref{equa:hybrid-lowerbound-evaluate} (line 10). It further computes  $\mathbb{I}^{+}_p(\mathcal{P}_{b}^{H^{opt}})$ using  $\mathcal{P}^{L^{*}}_{b}$ and $\mathcal{R}_1$ by Eq.~\ref{equa:hybrid-upperbound-evaluate} (line 11), and sets the approximation ratio to  $\mathbb{I}^{-}_p(\mathcal{P}^{H^{*}}_{b})$/ $\mathbb{I}^{+}_p(\mathcal{P}_{b}^{H^{opt}})$ (line 12). AP terminates and returns $\mathcal{P}^{H^{*}}_{b}$, if a ($1- 1/e-\epsilon$) approximation ratio is achieved (lines 13--14). Otherwise, it doubles the RR sets in $\mathcal{R}_1$ and $\mathcal{R}_2$ and perform a new iteration (line 15).

\textbf{Discussion.} Algorithm \ref{alg:approx} initializes parameters $\eta_{l}$ and $\eta_{u}$ in Eqs.~\ref{equa:hybrid-lowerbound-evaluate} and~\ref{equa:hybrid-upperbound-evaluate} as $\ln(3i_{\max}/\delta)$ (line 7), so parameters $\delta_{l}$ and $\delta_{u}$ in Lemma~\ref{lemma:poi-inf-upper-lower-inf} are fixed at $\delta/3i_{\max}$. Under this setting, the probability of deriving an incorrect approximation ratio in each iteration is at most $2\delta/3i_{\max}$. By the union bound for the first $(i_{\max}-1)$ iterations, the probability that Algorithm \ref{alg:approx} fails to find a qualified result is less than $2/3\delta$. Also, in the last iteration, by Lemma~\ref{lemma:rrset-size-bound},  Algorithm~\ref{alg:approx} ensures that $\mathcal{P}^{H^{*}}_{b}$ is a ($1- 1/e-\epsilon$) approximation result with a probability of ($1-\delta/3$) . Overall, AP guarantees a ($1- 1/e-\epsilon$) approximate solution with probability at least $1-\delta$.

\vspace{0mm}
\subsection{Heuristic Solution}
\label{sec:heuristic}

As an alternative solution to MaxInfBR$k$NN, we further develop a heuristic approach (\textbf{denoted as HE}) to trade accuracy for efficiency by introducing several optimizations in AP.

\vspace{1mm}
\textbf{Aggressive pruning and compensation.} When a POI in $\mathcal{P}_c$ contains some less frequent keywords, the conditions in line 18 of Algorithm \ref{algo:core_function} may be too strict to prune any pseudo-users. To enable pruning, instead of computing $\mathcal{L}^{\downarrow}_{\textit{SB}}(t)$ for each $t\in u_{ps}.key$, we only consider the keywords that occur frequently in users. In other words, keywords rarely contained by users are pruned aggressively. To ensure the correctness of batch results, users with such keywords are evaluated separately by Algorithm \ref{algo:usr_lbsl} for compensation. This modification improves the pruning efficiency of  Algorithm~\ref{algo:batch}.

\textbf{Assign priorities to BR$k$NN candidates.}
Since our problem only cares about most influential POIs, it is unnecessary to exactly retrieve BR$k$NNs for each $p\in\mathcal{P}_c$. Hence, we revise lines 14--17 of Algorithm \ref{algo:batch} as follows. Let $\hat{S}_r(p)$ ($\check{S}_r(p)$) be the set of candidate users (potential users) currently obtained (verified) for $p\in\mathcal{P}_c$. T$k$GSKQ verification is performed for each $u\in\mathcal{U}_c$ in descending order of user priority. The priority of a user $u$ is given as $|F(u)|\cdot|\mathcal{AL}(u)|$, since a candidate user $u$ associated with more friends and POIs is likely to be more influential. As more influential users are verified earlier in lines 14--17 of Algorithm \ref{algo:batch}, the difference between $\hat{S}_r(p)$ and $\check{S}_r(p)$ decreases, then we can find highly qualified results without verifying all $u\in\mathcal{U}_c$ with the strategies covered next.

\textbf{POI selection by approximating local influence.}
About 90\% of influence propagation occurs within 3 hops~\cite{tang2017asonam,tang2018snam}, . Thus, the 2-hop local influence can serve as a hint to obtain high-quality POI sets. Let $\hat{\mathcal{G}}_H$ and $\check{\mathcal{G}}_H$ be the heterogeneous graph built on $\hat{S}_r(p)$ and $\check{S}_r(p)$ for each $p\in\mathcal{P}_c$. Then, each time when a small batch of candidate users is verified, we perform the 2-hop based greedy algorithm on $\hat{\mathcal{G}}_H$ and $\check{\mathcal{G}}_H$ to get sets  $\hat{\mathcal{P}}_{b}^{L^{*}}$ and $\check{\mathcal{P}}_{b}^{L^{*}}$, respectively. If the accuracy satisfies $\mathbb{I}_{h}^{L}(\check{\mathcal{P}}_{b}^{L^{*}})/\mathbb{I}_{h}^{L}(\hat{\mathcal{P}}_{b}^{L^{*}})\geq1-1/e$, we return $\check{\mathcal{P}}_{b}^{L^{*}}$ as the result. $\check{\mathcal{P}}_{b}^{L^{*}}$ is expected to have good quality, due to the effectiveness of 2-hop based influence estimation~\cite{tang2018snam}. 

%% file: exp.tex
\section{Performance Study}
\label{sec:exp}
This section evaluates the proposed methods experimentally. All the methods are implemented in C++, and are tested on a PC with an Intel i7 2.90GHz CPU and 64GB RAM.

\subsection{Experimental Settings}
\label{setup}

\textbf{Datasets}. We employ three datasets: LasVegas, Gowalla, and Twitter. In particular, LasVegas is a real-world dataset obtained from the Yelp\footnote{http://www.yelp.com/dataset $\_$challenge/}, which contains POIs and users located in the city of Las Vegas. Each POI contains an ID, a textual description, and latitude and longitude. Each user contains check-ins, reviews of POIs, and IDs of  friends. The road network of LasVegas is obtained from OpenStreetMap\footnote{https://www.openstreetmap.org/}, and we map each POI and user to the closest vertex on the road network. Gowalla and Twitter are synthetic datasets generated by combining the road networks\footnote{http://www.dis.uniroma1.it/challenge9/} with geo-tagged social networks extracted from existing datasets of~\cite{li2014sigmod}. The keywords in Gowalla and Twitter are created to follow Zipf distributions~\cite{wu2012tkde}. Following the IM literature~\cite{huang2020vldbj,kempe2003kdd,li2014sigmod,tang2018sigmod}, we utilize the weighted cascade model~\cite{li2014sigmod,tang2015sigmod} to set the weight of each edge $(u,v)$ in $\mathcal{G}_s$ to $1/d_{in}(v)$, where $d_{in}(v)$ is the in-degree of user $v$ in $\mathcal{G}_s$.
The statistics of all the datasets are given in Table \ref{dataset-statistic}.

\begin{table}[t]
\centering
\begin{small}
\caption{Statistics of the datasets}
\label{dataset-statistic}
\vspace{-2mm}
\begin{tabular}{|p{3.92 cm}|p{1.1cm}<{\centering} |p{1.05cm}<{\centering} |p{1.1cm}<{\centering}  |}
\hline
&\textbf{LasVegas}&\textbf{Gowalla}&\textbf{Twitter} \\
\hline
\# of road network vertices            &  43,401       & 1,070,376      &1,890,815    \\ \hline
\# of road network edges           &    109,160     & 2,712,798      & 4,630,444    \\ \hline
\# of social network vertices            &    87,990     & 196,228      & 2,062,280    \\ \hline
\# of social network edges            &    552,486     & 1,900,654      & 4,803,440    \\ \hline
\# of POIs                       &   28,851       &  218,659    & 326,662    \\ \hline
\# of users(with reviews)                       &   28,867       &  196,228     &  1,114,261   \\ \hline
average \# of keywords per POI         &   350.2     &  8.1    & 4.4    \\ \hline
average \# of keywords per user         &   50.2     & 3.1   &  2.8     \\ \hline
average \# of check-ins per POI             &   63.3      &  5.0     &  5.3 \\ \hline
average \# of friends per user             &   6.3      &  9.7     &  8.7 \\ \hline
\end{tabular}
\end{small}
\vspace{-3mm}
\end{table}

\begin{table}[t]
\centering
\begin{small}
\caption{Parameter settings}
\label{paras}
\vspace{-2mm}
\begin{tabular}{|p{5.0cm}|p{3.0cm}<{\centering}|}
\hline
\textbf{Parameter}& \textbf{Range}    \\ \hline
$k$ (\# of top-$k$ ranking results)     &     10,  15, \textbf{20}, 25, 30           \\ \hline
$|\mathcal{P}_c|$ (\# of candidate POIs)     &     20, 40, \textbf{60}, 80, 100           \\ \hline
$b$ (the size of POI seed set)      & 1, 3, \textbf{5}, 7, 9             \\ \hline
$\alpha $ (parameter in the ranking function)      & 0.0, 0.2, 0.4, \textbf{0.6}, 0.8             \\ \hline
\hspace{-0.05in}  $|\mathcal{U}|$ (the cardinality of users)     & 0.6, 0.7, 0.8, 0.9, \textbf{1.0}         \\ \hline
\hspace{-0.05in} $|\mathcal{P}|$ (the cardinality of POIs)    & 0.6, 0.7, 0.8, 0.9, \textbf{1.0}          \\ \hline
\end{tabular}
\end{small}
\vspace{-4mm}
\end{table}

\textbf{Query set generation.}
We obtain the POIs in the query sets as follows. First, we randomly select two frequent keywords (such as 'bar' and 'shopping mall') from the dataset vocabulary. Then, for each keyword, we extract the group of POIs covering this keywords and eliminate those without any check-ins. Each group of POIs is further divided into two parts based on the popularity of POIs among users (whether the number of check-ins exceeds the average check-ins per POI). Finally, we have four clusters of POIs ($\mathcal{P}_{c1}$ through $\mathcal{P}_{c4}$), we randomly select $|\mathcal{P}_c|$ POIs from each cluster into $\mathcal{P}_c$, and with  $\mathcal{P}_{c1}$ (i.e., the most popular POI set) as the default query set.

\begin{table}[t]
\centering
\begin{small}
\caption{Index size \& construction time}
\label{tab:index}
\vspace{-2mm}
\begin{tabular}{|p{1.1cm}<{\centering}|p{1.4cm}<{\centering}|p{1.4cm}<{\centering}|p{1.35cm}<{\centering}|p{1.4cm}<{\centering}|}
\hline
\multirow{2}{*}{\textbf{Dataset}} & \multicolumn{2}{c|}{\textbf{GIM-Tree}} & \multicolumn{2}{c|}{\textbf{IG-NVD}} \\ \cline{2-5}
                                                & Size (MB)          & Time (sec)          & Size (MB)         & Time (sec)         \\ \hline
LasVegas                        & 112           & 63           & 60           & 56           \\ \hline
Gowalla                        & 4,533           & 4,390          & 2,633          & 3460          \\ \hline
Twitter                        & 42,438          & 49,633         & 3,913         & 5,868         \\ \hline
\end{tabular}
\end{small}
\vspace{-3mm}
\end{table}

\begin{table}[t]
\begin{small}
\caption{Effectiveness of MaxInfBR$k$NN on the LasVegas dataset}
\vspace{-2.5mm}
\label{case}
\begin{tabular}{|p{0.29cm}|p{1.08cm}<{\centering}|p{1.1cm}<{\centering}|p{1.3cm}<{\centering}|p{0.85cm}<{\centering}|p{1.7cm}<{\centering}  |}
\hline
  & \textit{Relevance} & \textit{Influencer} & \textit{MaxBRkNN} & \textit{Random} & \textit{MaxInfBRkNN} \\  \hline
$\mathcal{P}_{c1}$ & 118  & 201 & 147 & 60& \textbf{255} \\ \hline
$\mathcal{P}_{c2}$ & 117   & 238 & 191 & 63 & \textbf{270}  \\ \hline
$\mathcal{P}_{c3}$ & 212   & 385 & 352 & 101 & \textbf{450}  \\ \hline
$\mathcal{P}_{c4}$ & 204   & 296 & 277 & 72 & \textbf{447}  \\ \hline
\end{tabular}
\vspace{-3mm}
\end{small}
\end{table}

\textbf{Competitors.}
We first compare the performance of our index (i.e., IG-NVD) with the existing state-of-the-art index (i.e., the GIM-Tree)~\cite{zhao2017icde}. {Then, to investigate the effectiveness of our problem in finding highly influential POIs, we compare the POI sets retrieved by MaxInfBR$k$NN queries with those retrieved by several simple but smart POI selection policies:}
\begin{itemize}\setlength{\itemsep}{-\itemsep}
\item{} \textbf{Relevance.} This policy greedily selects $b$ POIs from $\mathcal{P}_c$ having textual or social relevance to the largest number of users located in the local area (i.e., within 5$km$ away).
\item{}  \textbf{Influencer.} This policy selects the POIs whose BR$k$NNs include the largest number of influencers. Specifically, it disregards users without textual or social relevance to $\mathcal{P}_c$, and utilizes existing IM algorithms~\cite{tang2018sigmod,tang2014sigmod} to find top-$x$ (i.e., $x=200$) users as influencers. Then, it computes the top-$k$ results for each influencer, and selects $b$ POIs which are top-$k$ relevant to the largest number of influencers.
\item{} \textbf{MaxBR$k$NN}. This policy maximizes the cardinality of BR$k$NNs for any size-$b$ POI set. We revise an existing algorithm~\cite{choudhury2016vldb} by considering road network distance and social relevance in the similarity function, and greedily select $b$ POIs having the largest size of BR$k$NN results.

\item{} \textbf{Random.} As the most intuitive competitor, this policy simply selects $b$ random POIs from $\mathcal{P}_c$.
\end{itemize}

\textbf{Parameter Settings.} The parameter settings are listed in Table \ref{paras}, with default values in bold. Note that in Table \ref{paras}, $|\mathcal{U}|$ ($|\mathcal{P}|$) denotes the percentage of all users (POIs) w.r.t. the  number of users (POIs) for testing. In BA and AP, a trade-off between accuracy and efficiency is controlled by parameters $\epsilon$ and $\delta$, respectively. Following settings extensively used in recent studies~\cite{guo2020sigmod,yuchen2015vldb,tang2018sigmod}, we set $\epsilon$ and $\delta$ in Eq.~\ref{lemma:rrset-size-bound} to 0.1 and $1/|\mathcal{V}_s|$, respectively.

\begin{figure*}[t]
\centering
\includegraphics[width=0.7\textwidth]{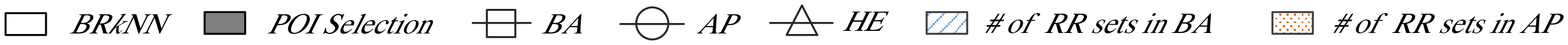}\\
  \hspace{ -1 mm}
  \subfigure[Total Runtime]{
   \label{effect-a} \vspace{-5mm}
   \raisebox{-0.5cm}{\includegraphics[width=1.4in]{{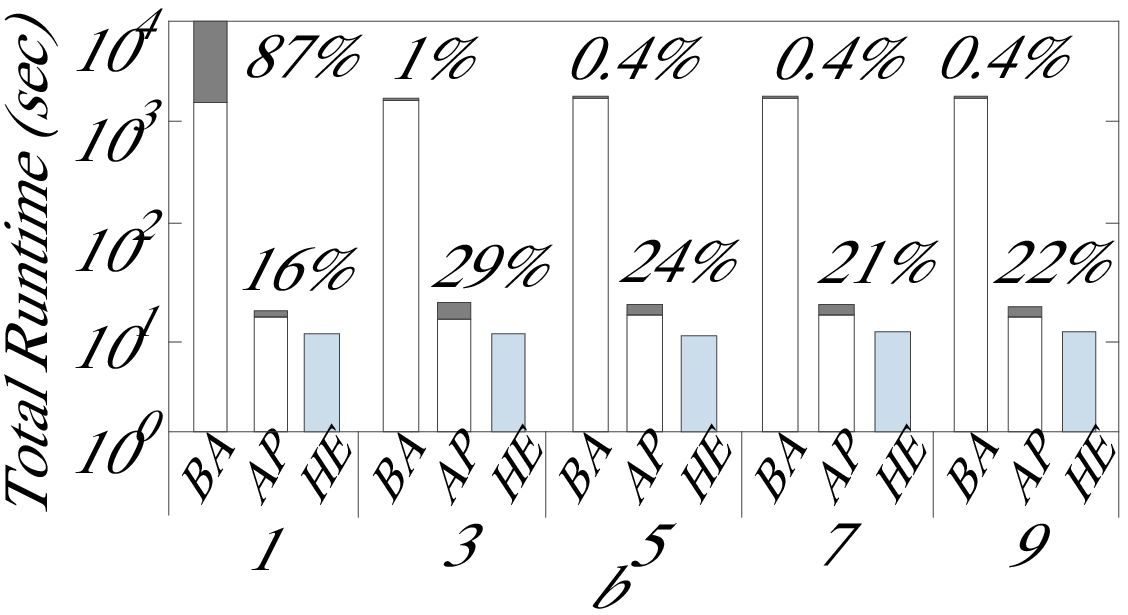}}} 
  }\hspace{ 1 mm}
  \subfigure[Influence of $\mathcal{P}_s$]{
   \label{effect-b} \vspace{-5mm}
   \raisebox{-0.5cm}{\includegraphics[width=1.08 in]{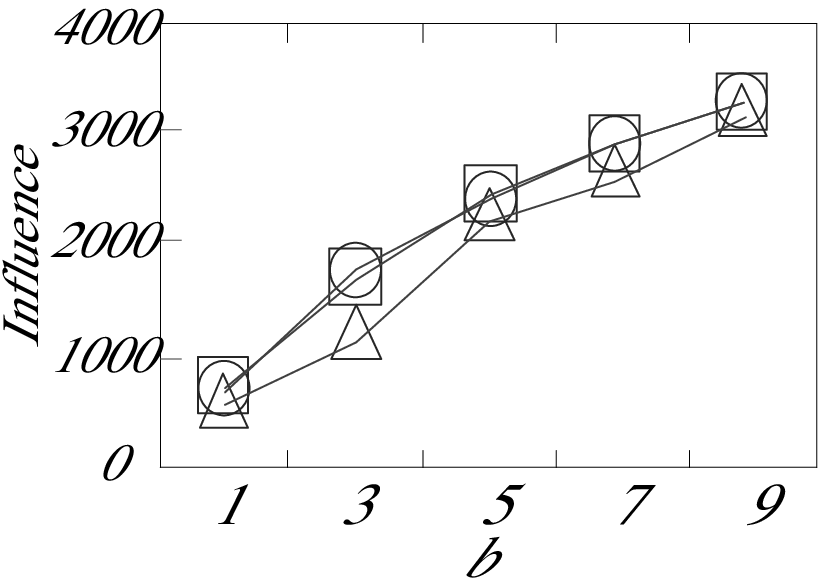}}
  }\hspace{ 1 mm}
  \subfigure[BR$k$NN Retrieval]{
   \label{effect-a} \vspace{-5mm}
   \raisebox{-0.5cm}{\includegraphics[width=1.08 in]{{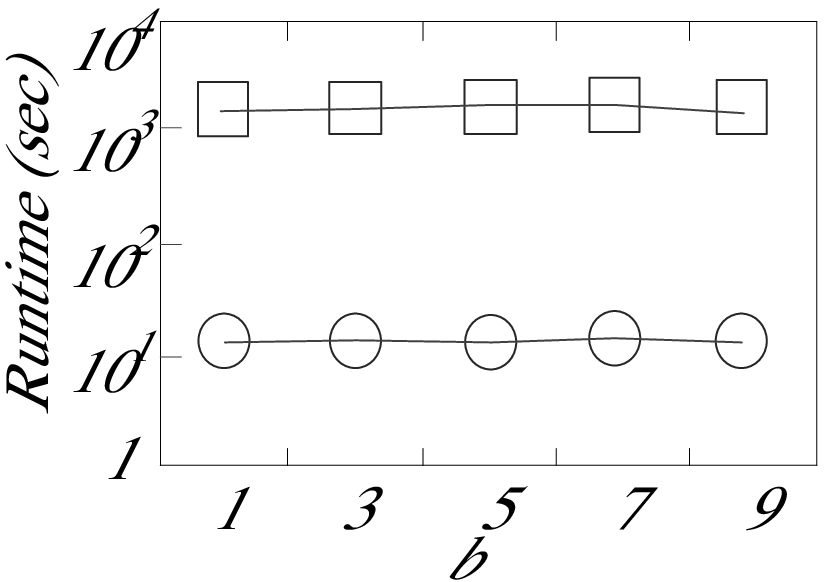}}}
  }\hspace{ 1 mm}
  \subfigure[POI Selection]{
   \label{effect-b} \vspace{-5mm}
   \raisebox{-0.5cm}{\includegraphics[width=1.25in]{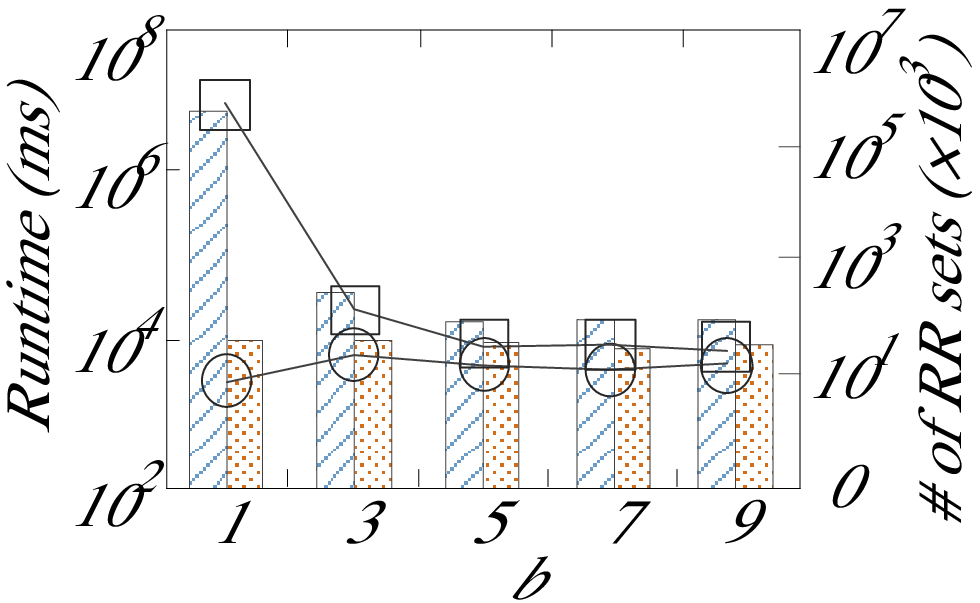}}
  }\hspace{ 1 mm}
  \subfigure[Approximation ratio]{
   \label{effect-b} \vspace{-5mm}
   \raisebox{-0.5cm}{\includegraphics[width=1.1 in]{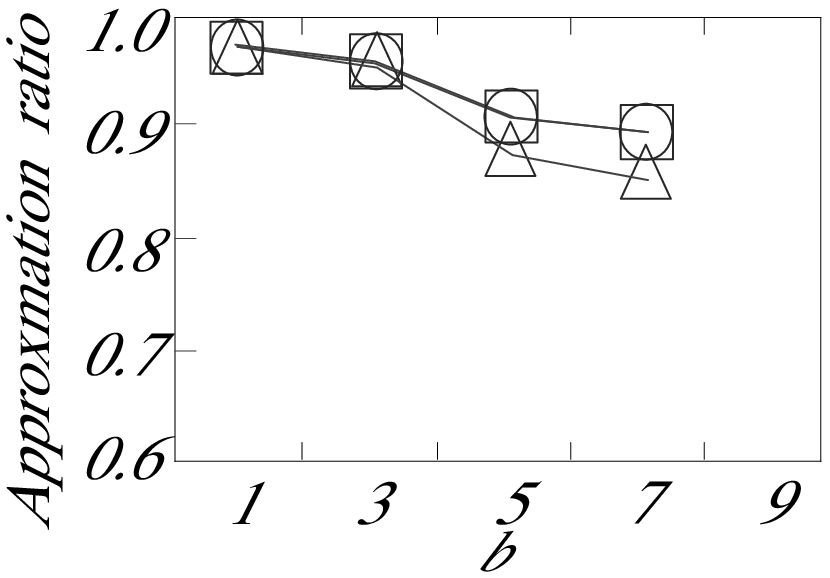}}
  }
\vspace{-2.7mm}
\caption{Performance comparison when varying $b$ on the Gowalla dataset}
\label{fig:varying_b}
\vspace{-2.2mm}
\end{figure*}

\begin{figure*}[t]
\centering
  \hspace{ -1 mm}
  \subfigure[Total Runtime]{
   \label{effect-a} \vspace{-5mm}
   \raisebox{-0.5cm}{\includegraphics[width=1.4in]{{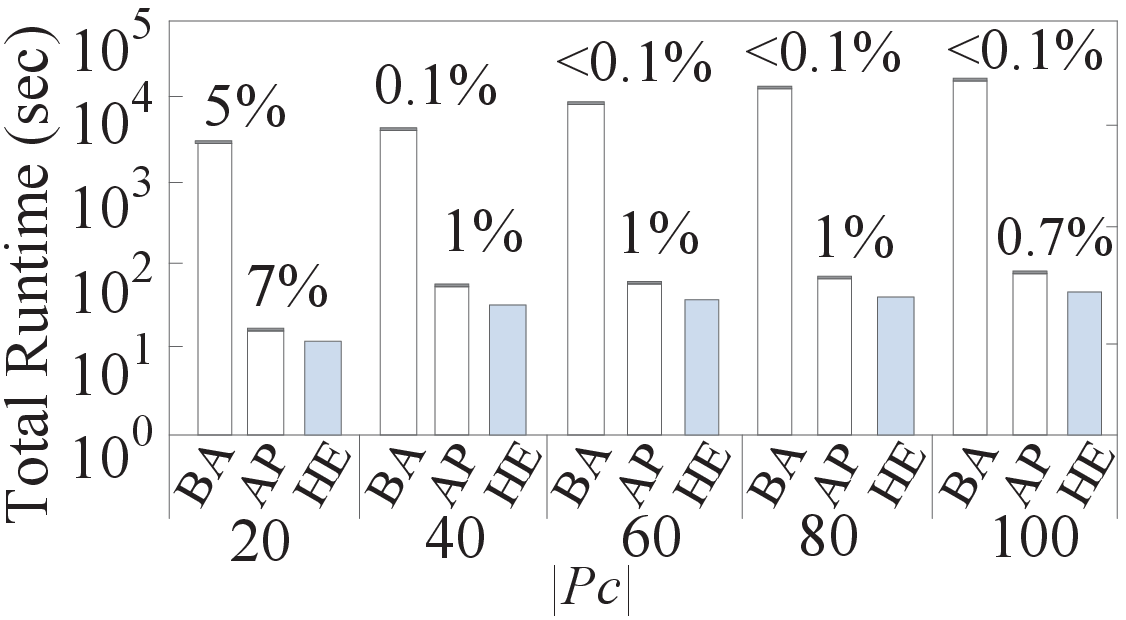}}} 
  }\hspace{ 1 mm}
  \subfigure[Influence of $\mathcal{P}_s$]{
   \label{effect-b} \vspace{-5mm}
   \raisebox{-0.5cm}{\includegraphics[width=1.1 in]{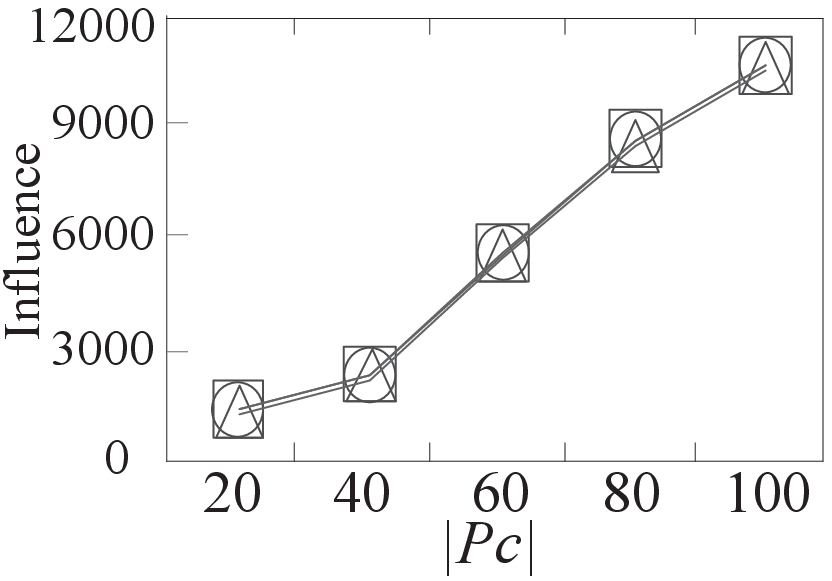}}
  }\hspace{ 1 mm}
  \subfigure[BR$k$NN Retrieval]{
   \label{effect-a} \vspace{-5mm}
   \raisebox{-0.5cm}{\includegraphics[width=1.08 in]{{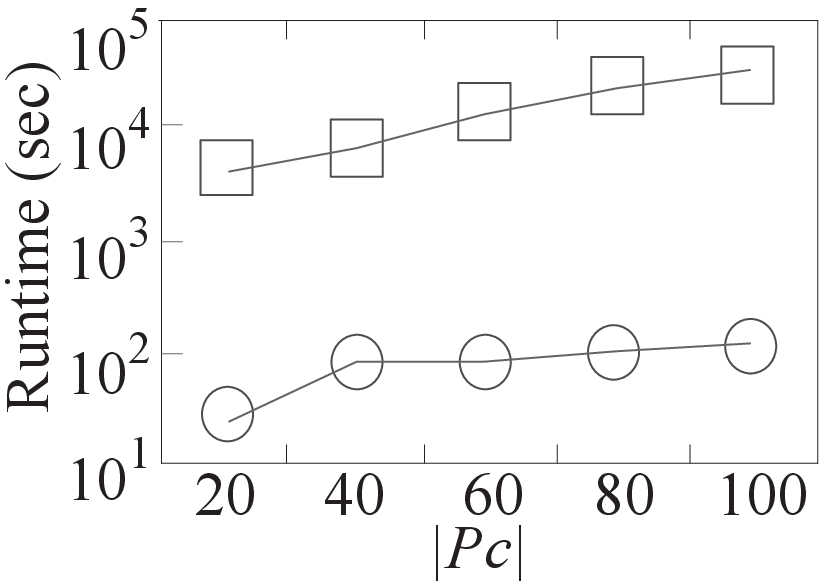}}}
  }\hspace{ 1 mm}
  \subfigure[POI Selection]{
   \label{effect-b} \vspace{-5mm}
   \raisebox{-0.5cm}{\includegraphics[width=1.25in]{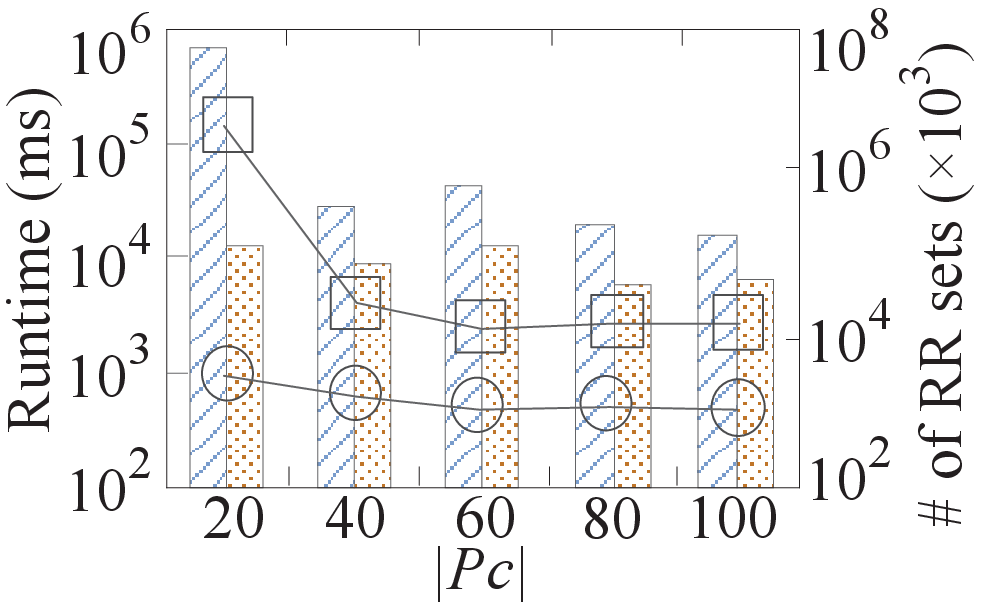}}
  }\hspace{ 1 mm}
  \subfigure[Approximation ratio]{
   \label{effect-b} \vspace{-5mm}
   \raisebox{-0.5cm}{\includegraphics[width=1.1 in]{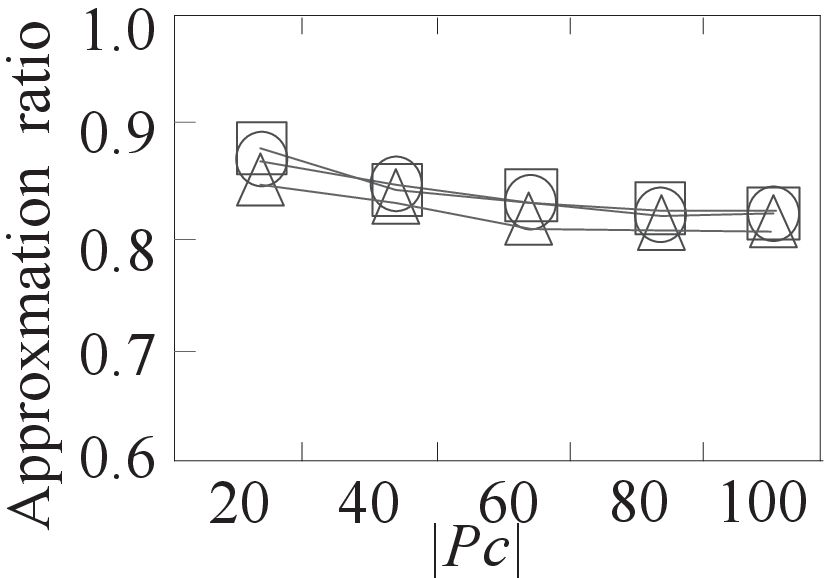}}
  }
\vspace{-2.7mm}
\caption{Performance comparison when varying $|\mathcal{P}_c|$ on the Twitter  dataset}
\label{fig:varying_Pc}
\vspace{-2.5mm}
\end{figure*}

\begin{figure*}[t]
\centering
  \hspace{ -1 mm}
  \subfigure[Total Runtime]{
   \label{effect-a} \vspace{-5mm}
   \raisebox{-0.5cm}{\includegraphics[width=1.4in]{{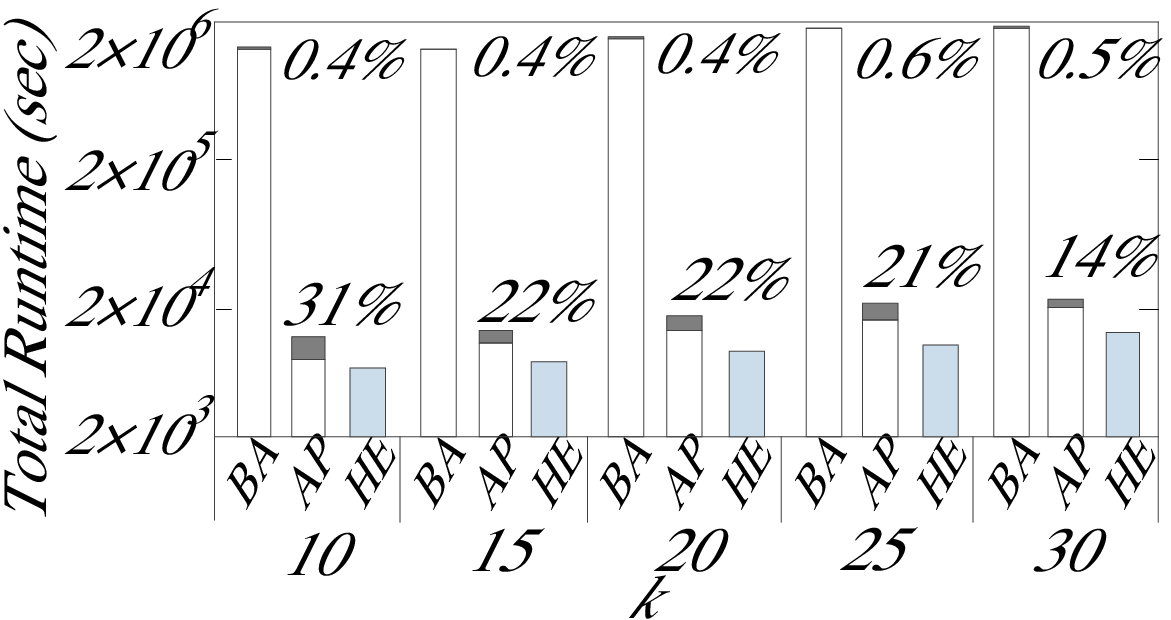}}} 
  }\hspace{ 1 mm}
  \subfigure[Influence of $\mathcal{P}_s$]{
   \label{effect-b} \vspace{-5mm}
   \raisebox{-0.5cm}{\includegraphics[width=1.08 in]{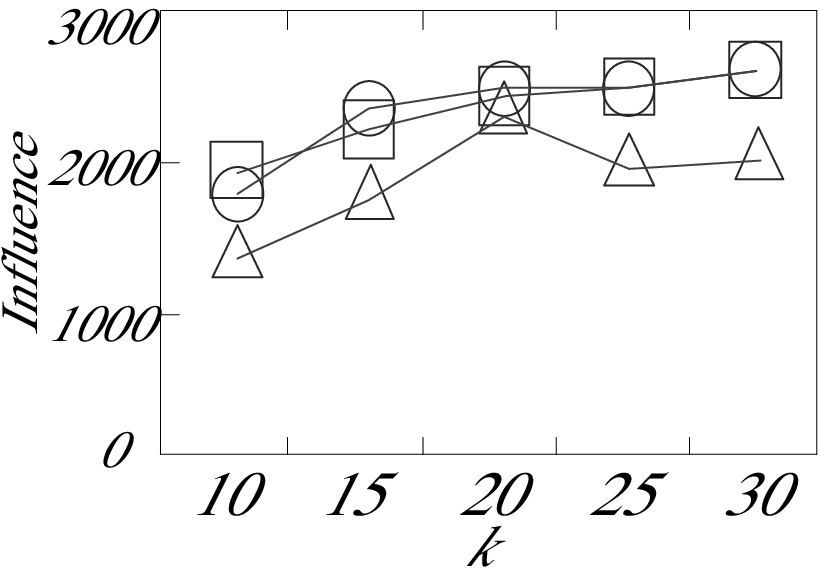}}
  }\hspace{ 1 mm}
  \subfigure[BR$k$NN Retrieval]{
   \label{effect-a} \vspace{-5mm}
   \raisebox{-0.5cm}{\includegraphics[width=1.08 in]{{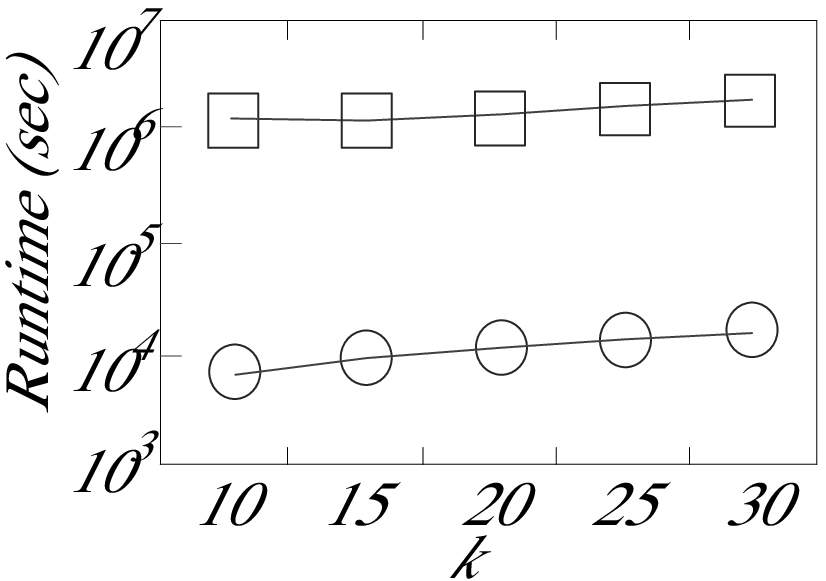}}}
  }\hspace{ 1 mm}
  \subfigure[POI Selection]{
   \label{effect-b} \vspace{-5mm}
   \raisebox{-0.5cm}{\includegraphics[width=1.25in]{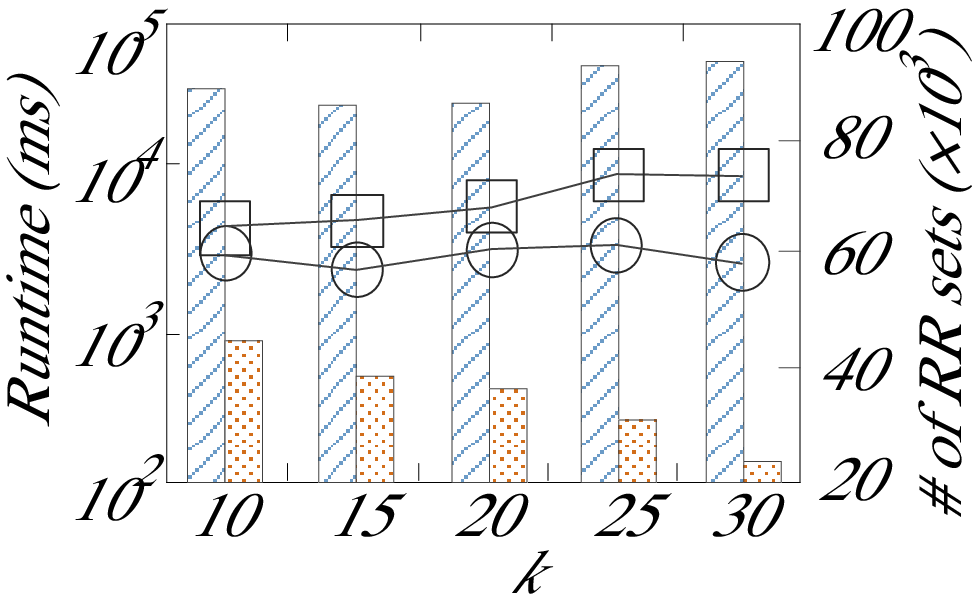}}
  }\hspace{ 1 mm}
  \subfigure[Approximation ratio]{
   \label{effect-b} \vspace{-5mm}
   \raisebox{-0.5cm}{\includegraphics[width=1.1 in]{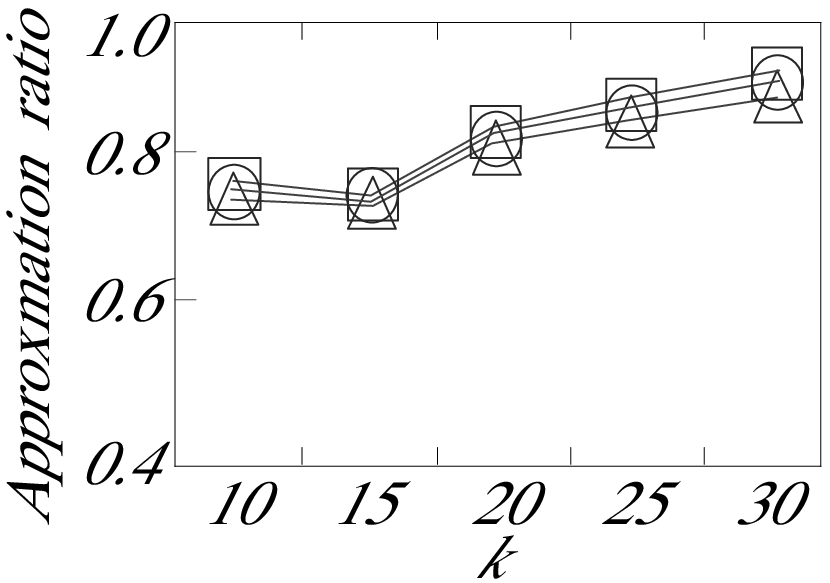}}
  }
\vspace{-2.7mm}
\caption{Performance comparison when varying $k$ on the Gowalla dataset}
\label{fig:varying_k}
\vspace{-4mm}
\end{figure*}

\subsection{Indexing Performance}

We compare the performance of our index (i.e., IG-NVD) with existing indexes (i.e., GIM-Tree~\cite{zhao2017icde}) for retrieving BR$k$NNs on geo-social networks. Statistics in Table~\ref{tab:index} illustrate that, our index needs much less storage space than the GIM-Tree when the dataset becomes large. Specifically, the index size of the GIM-Tree is nearly 5 GB on Gowalla and increases to more than 40 GB on the largest dataset, Twitter, which may not fit in the main memory on a commodity computer. 
In contrast, our index is only about 4 GB on Twitter, with a construction time less than 2 hours. The two indexes have similar construction time on the smallest dataset (i.e., LasVegas). This is because there are few users and POIs in LasVegas, thus making the pre-processing costs on social information indexed in GIM-Tree relatively low.

\vspace{-2mm}
\subsection{Effectiveness of POI Selection}
\label{exp:effective-POI}

We report the influence of $\mathcal{P}_s$ returned by the different POI selection policies using the real dataset LasVegas. As shown in Table \ref{case}, for different kinds of POI candidate sets in LasVegas, MaxInfBR$k$NN consistently outperforms the competitors. Also, Influencer performs the best among the competitors, which indicates the importance of taking into account social influencers. However, the gap between Influencer and MaxInfBR$k$NN is large, because Influencer only selects POIs affecting top influencers (top-$x$ influential users), but ignores users with relatively less influence, which leads to considerable influence loss. In contrast, MaxInfBR$k$NN detects all the users in BR$k$NNs and comprehensively evaluates their influence; thus, MaxInfBR$k$NN can find highly  qualified results. 

\subsection{Efficiency and Scalability Evaluation}
\label{exp:perform}

We further utilize larger synthetic datasets (i.e., Gowalla and Twitter) to comprehensively evaluate the scalability and efficiency of our proposed methods (i.e., BA, AP, and HE presented in Sections~\ref{sec:preli},~\ref{subsec:AP}, and~\ref{sec:heuristic}, respectively). In each set of experiment, we only vary a single parameter while fixing the remaining ones at their defaults.

We evaluate the impact of different parameter settings on our methods, in terms of i) the total runtime, ii) the influence of retrieved POI set (i.e., $\mathbb{I}_p(\mathcal{P}_s)$), iii) the BR$k$NN retrieval time, iv) the overhead  incurred in POI selection and qualification (i.e., POI selection time and the number of sampled RR sets), v) the approximation ratio of our methods against the exact solution. Recall that the total runtime of BA and AP consists of the BR$k$NN retrieval time and the POI selection time. The numbers at the top of the columns in Figs.~\ref{fig:varying_b}(a)--\ref{fig:varying_P}(a) denote the percentage of the POI selection time in the total runtime. We do not report the number of sampled RR sets for HE, since HE utilizes heuristics instead of RR sets to select POIs.

\begin{figure*}[t]
\centering
\includegraphics[width=0.7\textwidth]{performance_icon.eps}\\
  \hspace{ -1 mm}
  \subfigure[Total Runtime]{
   \label{effect-a} \vspace{-5mm}
   \raisebox{-0.5cm}{\includegraphics[width=1.4in]{{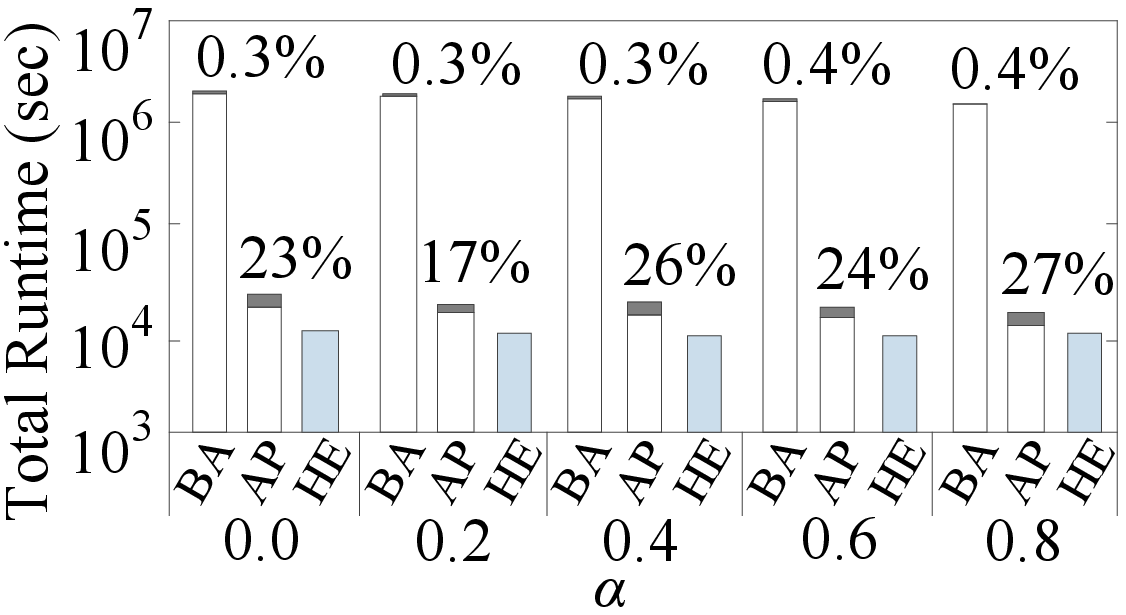}}} 
  }\hspace{ 1 mm}
  \subfigure[Influence of $\mathcal{P}_s$]{
   \label{effect-b} \vspace{-5mm}
   \raisebox{-0.5cm}{\includegraphics[width=1.08 in]{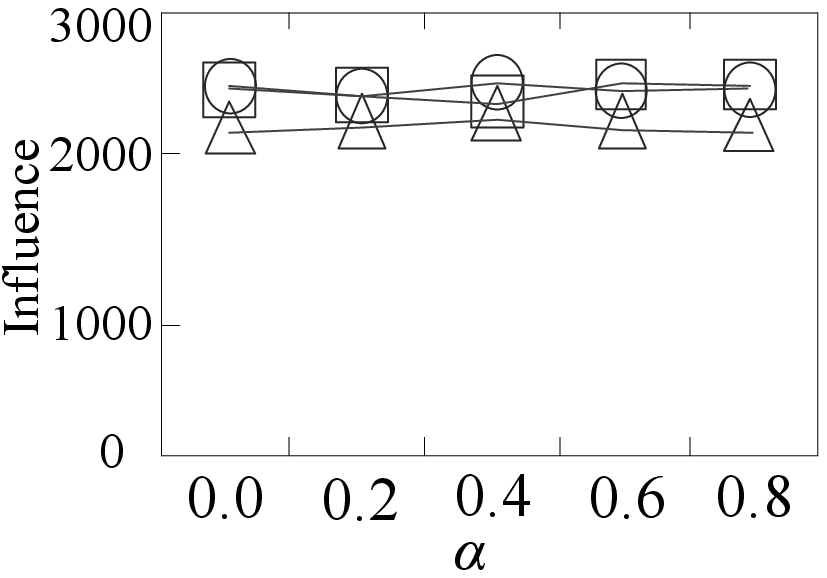}}
  }\hspace{ 1 mm}
  \subfigure[BR$k$NN Retrieval]{
   \label{effect-a} \vspace{-5mm}
   \raisebox{-0.5cm}{\includegraphics[width=1.08 in]{{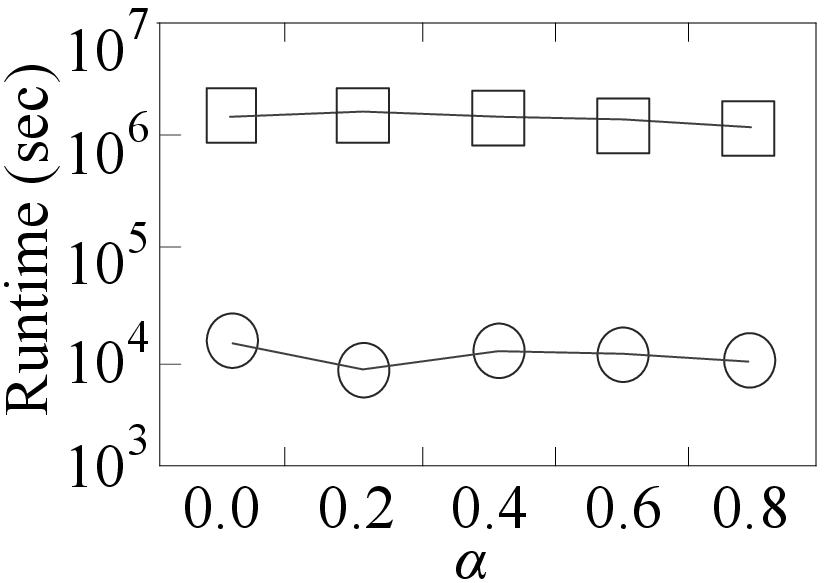}}}
  }\hspace{ 1 mm}
  \subfigure[POI Selection]{
   \label{effect-b} \vspace{-5mm}
   \raisebox{-0.5cm}{\includegraphics[width=1.25in]{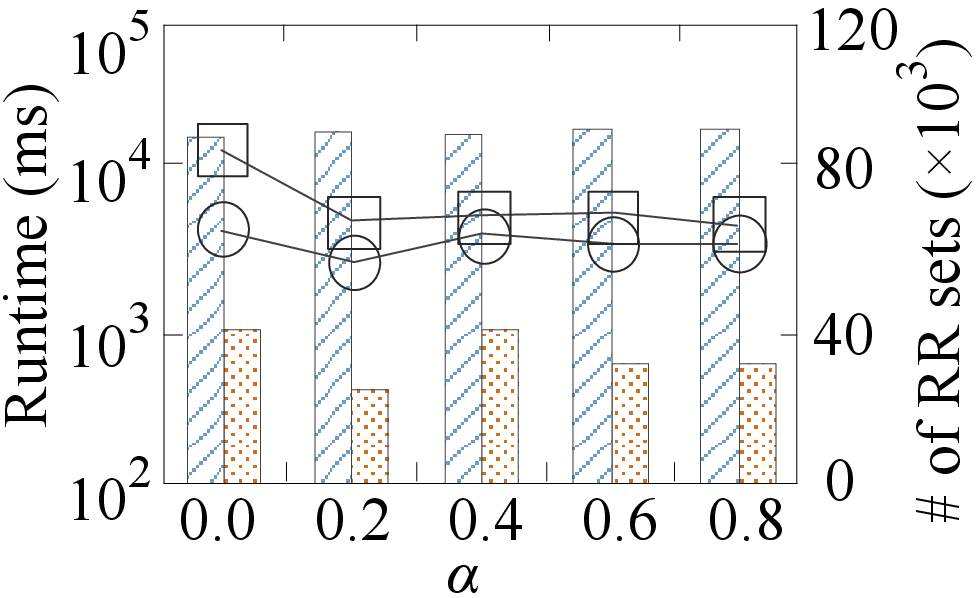}}
  }\hspace{ 1 mm}
  \subfigure[Approximation ratio]{
   \label{effect-b} \vspace{-5mm}
   \raisebox{-0.5cm}{\includegraphics[width=1.1 in]{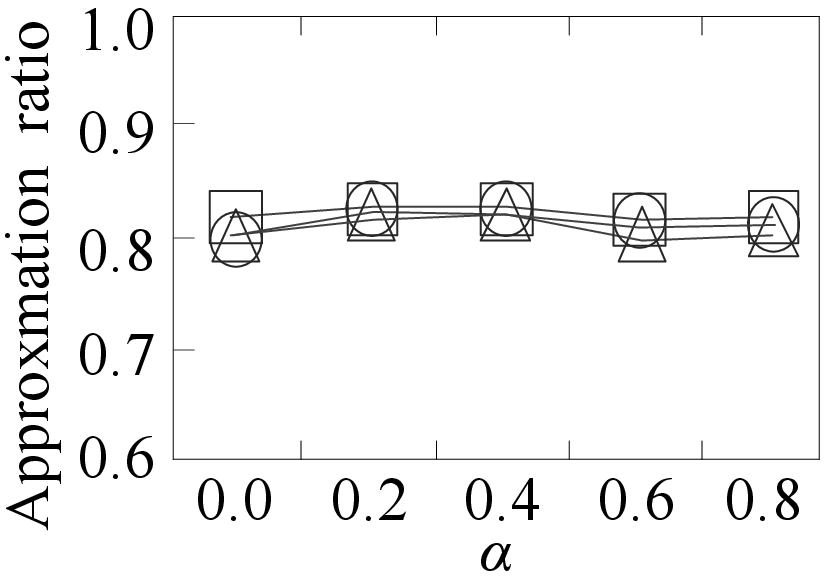}}
  }
\vspace{-2.7mm}
\caption{Performance comparison when varying $\alpha$ on the Gowalla dataset}
\label{fig:varying_alpha}
\vspace{-2.2mm}
\end{figure*}

\begin{figure*}[t]
\centering
  \hspace{ -1 mm}
  \subfigure[Total Runtime]{
   \label{effect-a} \vspace{-5mm}
   \raisebox{-0.5cm}{\includegraphics[width=1.4in]{{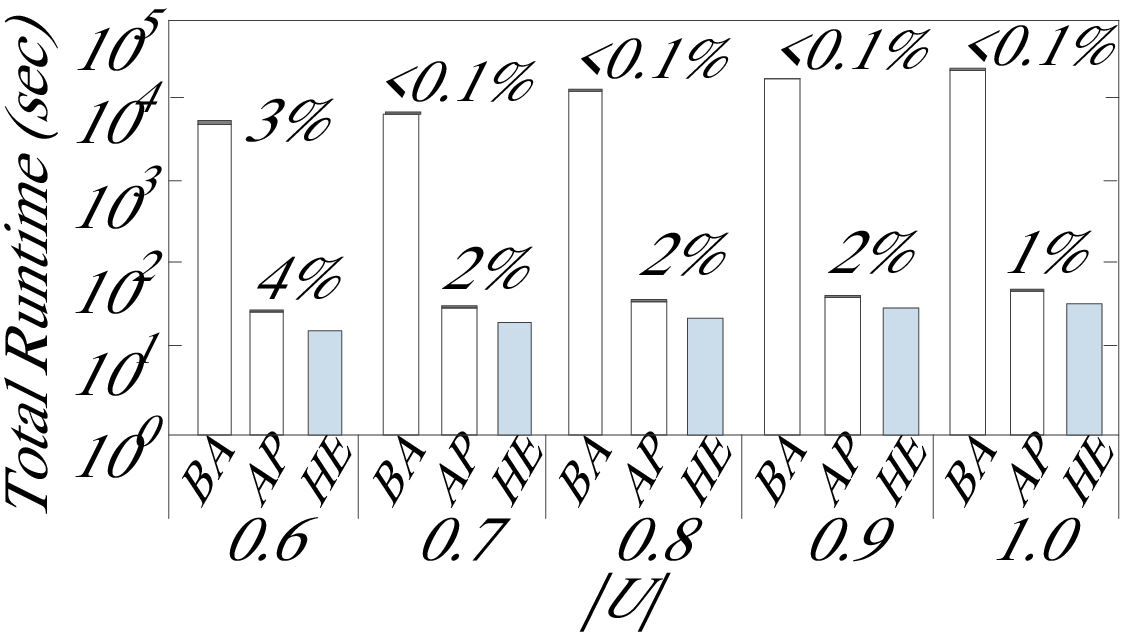}}} 
  }\hspace{ 1 mm}
  \subfigure[Influence of $\mathcal{P}_s$]{
   \label{effect-b} \vspace{-5mm}
   \raisebox{-0.5cm}{\includegraphics[width=1.08 in]{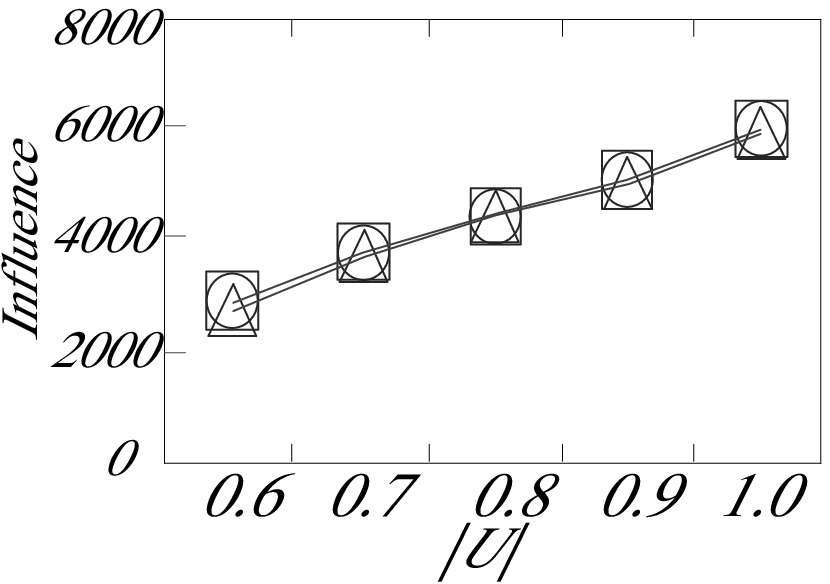}}
  }\hspace{ 1 mm}
  \subfigure[BR$k$NN Retrieval]{
   \label{effect-a} \vspace{-5mm}
   \raisebox{-0.5cm}{\includegraphics[width=1.08 in]{{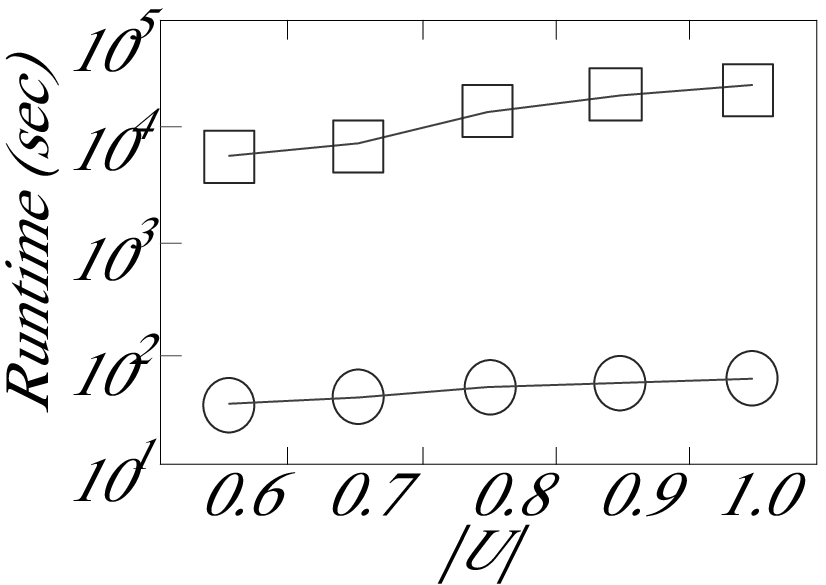}}}
  }\hspace{ 1 mm}
  \subfigure[POI Selection]{
   \label{effect-b} \vspace{-5mm}
   \raisebox{-0.5cm}{\includegraphics[width=1.25in]{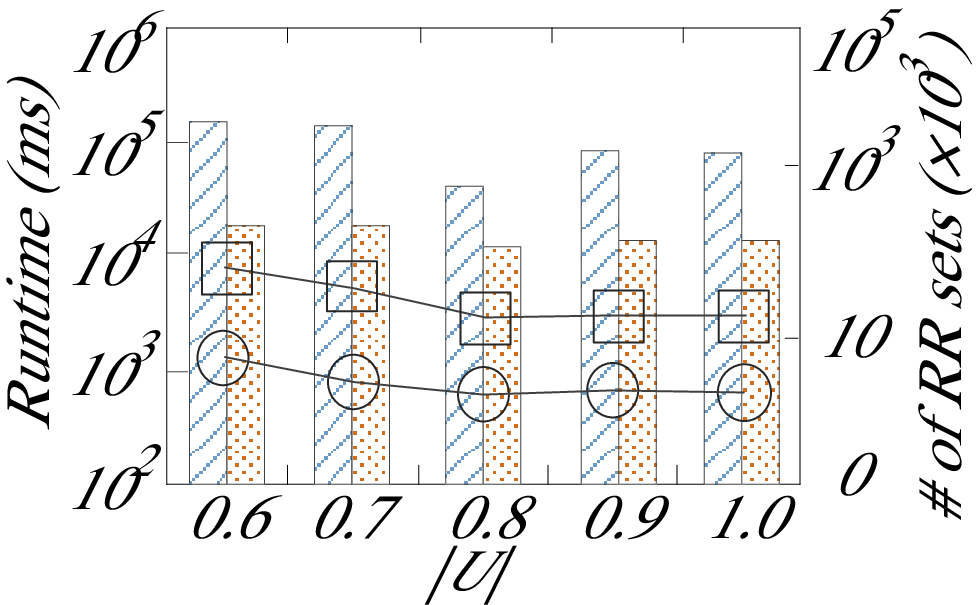}}
  }\hspace{ 1 mm}
  \subfigure[Approximation ratio]{
   \label{effect-b} \vspace{-5mm}
   \raisebox{-0.5cm}{\includegraphics[width=1.1 in]{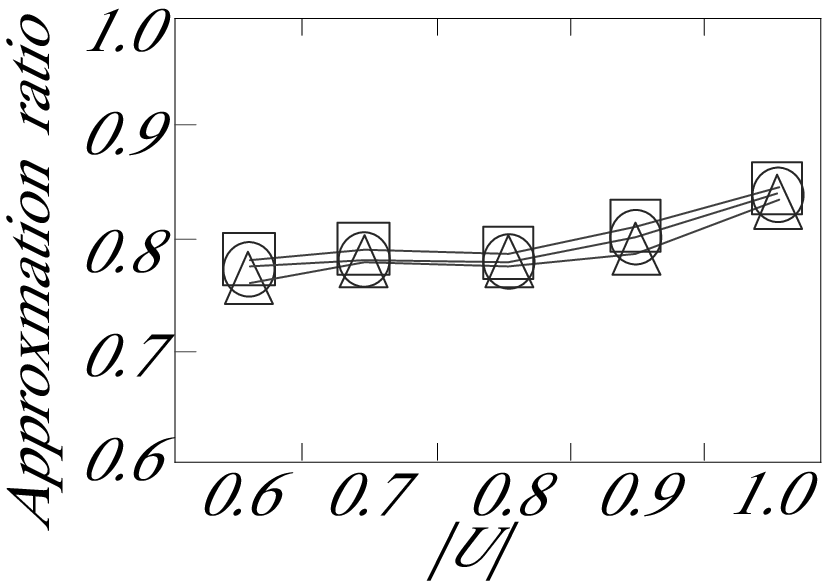}}
  }
\vspace{-2.7mm}
\caption{Performance comparison when varying $|\mathcal{U}|$ on the Twitter  dataset}
\label{fig:varying_U}
\vspace{-2.5mm}
\end{figure*}

\begin{figure*}[t]
\centering
  \hspace{ -1 mm}
  \subfigure[Total Runtime]{
   \label{effect-a} \vspace{-5mm}
   \raisebox{-0.5cm}{\includegraphics[width=1.4in]{{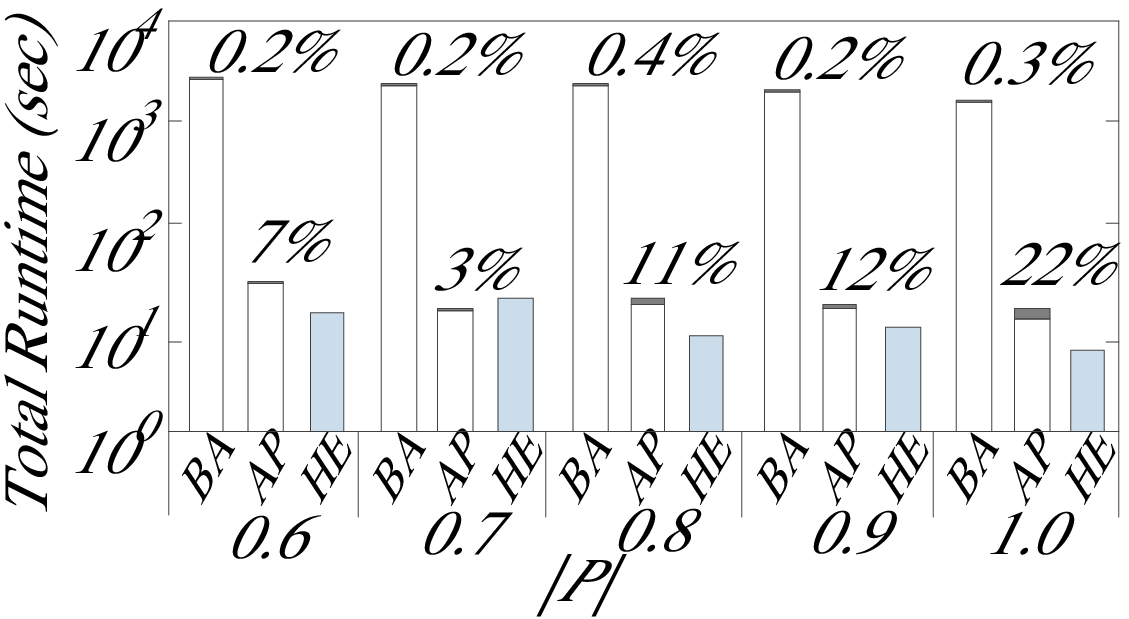}}} 
  }\hspace{ 1 mm}
  \subfigure[Influence of $\mathcal{P}_s$]{
   \label{effect-b} \vspace{-5mm}
   \raisebox{-0.5cm}{\includegraphics[width=1.08 in]{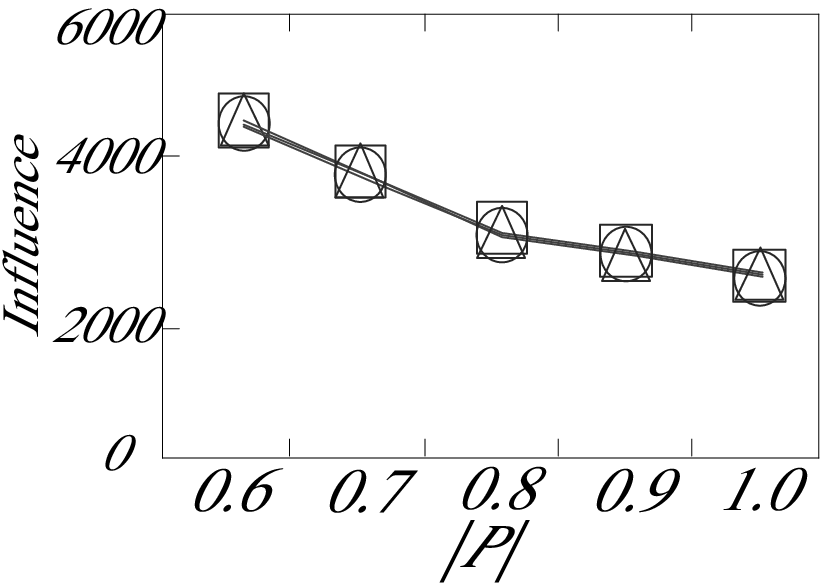}}
  }\hspace{ 1 mm}
  \subfigure[BR$k$NN Retrieval]{
   \label{effect-a} \vspace{-5mm}
   \raisebox{-0.5cm}{\includegraphics[width=1.08 in]{{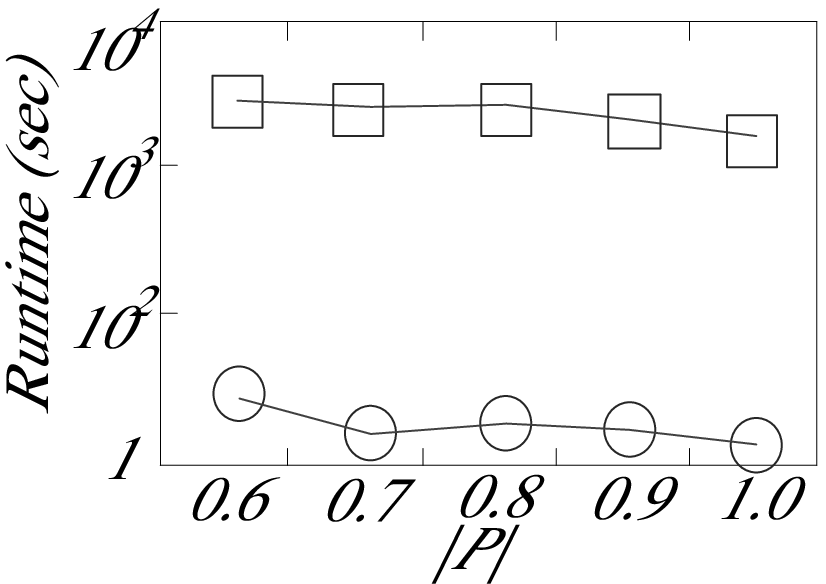}}}
  }\hspace{ 1 mm}
  \subfigure[POI Selection]{
   \label{effect-b} \vspace{-5mm}
   \raisebox{-0.5cm}{\includegraphics[width=1.25in]{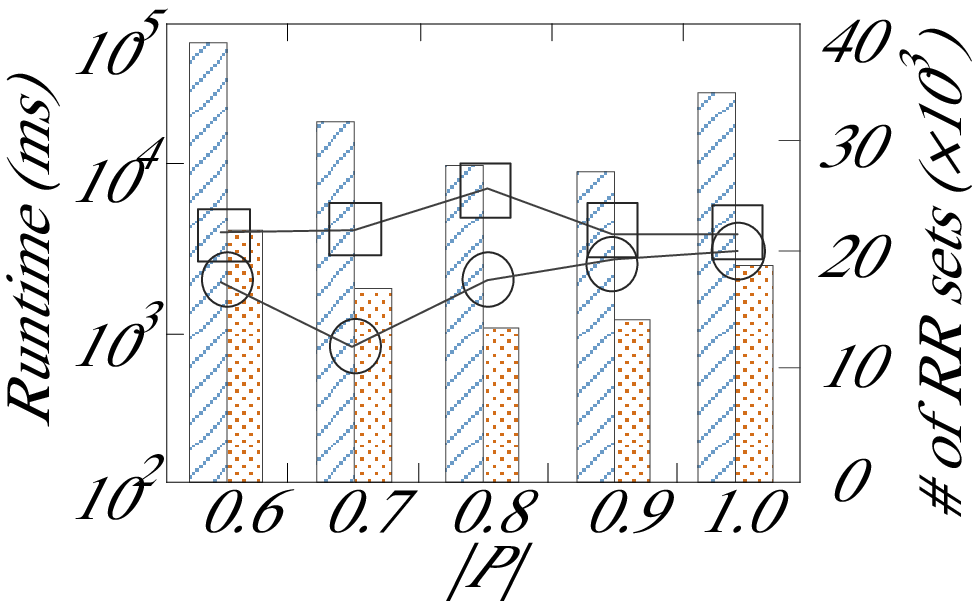}}
  }\hspace{ 1 mm}
  \subfigure[Approximation ratio]{
   \label{effect-b} \vspace{-5mm}
   \raisebox{-0.5cm}{\includegraphics[width=1.1 in]{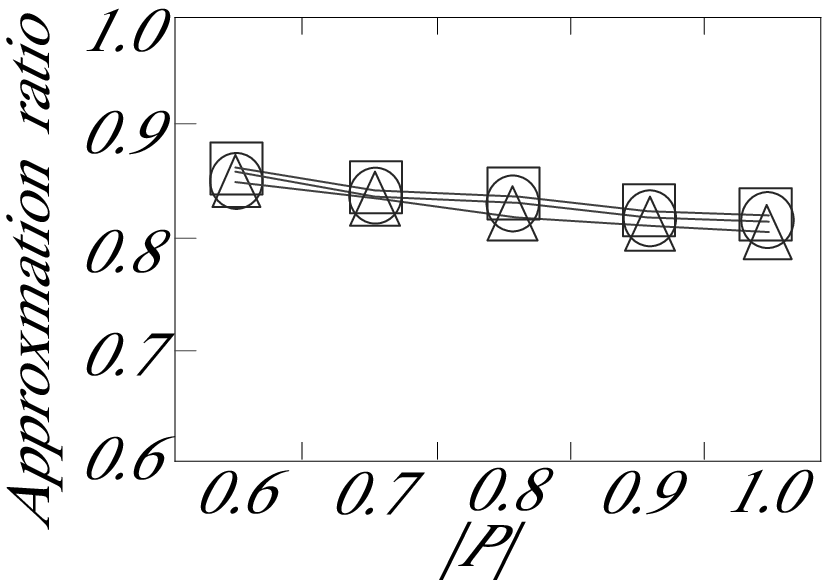}}
  }
\vspace{-2.7mm}
\caption{Performance comparison when varying $|\mathcal{P}|$ on the Gowalla  dataset}
\label{fig:varying_P}
\vspace{-3.5mm}
\end{figure*}

\emph{Effect of $b$}. The experimental results of our methods when varying $b$ are shown in Fig.~\ref{fig:varying_b}.  We observe that, although BA, AP, and HE achieve similar influence (cf.\ Fig.~\ref{fig:varying_b}(b)), the running overhead of BA is significantly larger than those of AP and HE. Also, the performance of BA is not robust to $b$. When $b=1$, the total runtime of BA increases sharply (cf.\ Fig.~\ref{fig:varying_b}(a)), and nearly a billion RR sets are sampled and maintained by BA (cf.\ Fig.~\ref{fig:varying_b}(d)). The reason is that BA only utilizes the coverage of RR sets to estimate the influence. Thus, when $b$ is small (i.e., $b=1$), relative estimation errors in BA grows, which forces BA to sample more RR sets to ensure the accuracy of the results. In contrast, AP and HE are more efficient and robust across different values of $b$. The BR$k$NN retrieval in AP is consistently 2 orders of magnitude faster than that in BA (cf.\ Fig.~\ref{fig:varying_b}(c)). The runtime of BR$k$NN retrieval in both BA and AP are not affected by $b$ since BR$k$NN computation does not change much when $|\mathcal{P}_c|$ is fixed. The influence of each method increases with $b$ while the approximation ratio drops (cf.\ Fig.~\ref{fig:varying_b}(e)). This occurs because more POIs are included in $|\mathcal{P}_s|$, which magnifies the errors in selecting POIs greedily.

\emph{Effect of $|\mathcal{P}_c|$}. Fig.~\ref{fig:varying_Pc} depicts the results when changing $|\mathcal{P}_c|$. The total runtime and the influence of each method increase with $|\mathcal{P}_c|$ (cf.\ Figs.~\ref{fig:varying_Pc}(a) and \ref{fig:varying_Pc}(b)), since more relevant users are evaluated during the process.
Note that the runtime of BR$k$NN retrieval in BA is prohibitively large on the Twitter dataset (more than $10$ hours), and increases sharply with the growth of $|\mathcal{P}_c|$ (cf.\ Fig.~\ref{fig:varying_Pc}(c)). In comparison, the BR$k$NN retrieval in AP runs much faster (only several minutes), and scales well as $|\mathcal{P}_c|$ grows. This illustrates the deficiency of existing methods~\cite{zhao2017icde}, and confirms that our batch processing algorithms are more scalable on larger datasets. As indicated by the percentages on the top of the columns in Fig.~\ref{fig:varying_Pc}(a), the POI selection time in both BA and AP accounts for less than 5\% in most cases, implying that the BR$k$NN retrieval is the query processing bottleneck.
In Fig.~\ref{fig:varying_Pc}(d), again, we find that the performance of BA is not robust and is sensitive to $|\mathcal{P}_c|$, while AP is relatively robust. The approximation ratio of our methods decreases slightly with $|\mathcal{P}_c|$, as more POIs are evaluated which increases the POI selection errors.

\emph{Effect of $k$}. Fig.~\ref{fig:varying_k} shows the results when varying $k$. As expected, AP and HE consistently outperform BA in terms of runtime, and all the methods achieve comparable influence. In particular, with the growth of $k$, the influence of the returned POI set as well as their approximation ratios ascend for each method (cf.\ Figs.~\ref{fig:varying_k}(b) and \ref{fig:varying_k}(e)). This is because, when $k$ becomes larger, more users are influenced by the POIs in $\mathcal{P}_c$, which increases influence and the accuracy of $\mathcal{P}_{s}$.

\emph{Effect of $\alpha$}. Fig.~\ref{fig:varying_alpha} illustrates the results when varying $\alpha$. Recall that  $\alpha$ controls the balance between textual similarity and social relevance in the scoring function, and that a higher value of $\alpha$ implies a higher preference for the social relevance. As expected, AP and HE consistently outperform BA in terms of running efficiency across different values of $\alpha$. The influence and the approximation ratio of each method are affected by $\alpha$ only little, because the largest influence achieved by any size-$b$ POI set only changes little when changing $\alpha$.

\emph{Effect of $|\mathcal{U}|$}. Fig.~\ref{fig:varying_U} depicts the results when varying the cardinality of users in $\mathcal{U}$. It is seen that the runtime of BA increases with $|\mathcal{U}|$, while the runtime of AP and HE are relatively stable. This shows the effectiveness of our pruning techniques, and the superiority of AP and HE in tackling large datasets. The influence and the approximation ratio of each method increase with $|\mathcal{U}|$, because more users are influenced by $\mathcal{P}_s$, which in turn enhances the accuracy of our algorithm when selecting influential POIs.

\emph{Effect of $|\mathcal{P}|$}. Fig.~\ref{fig:varying_P} plots the results when varying the cardinality of the POIs in $\mathcal{P}$.
As $|\mathcal{P}|$ grows, the running time of each method drops. This is because the top-$k$ score of each user increases with the growth of $\mathcal{P}$, enhancing the pruning power in all the methods. In addition, the influence and the approximation ratio of our methods drop as $|\mathcal{P}|$ grows. The reason is that, when $|\mathcal{P}|$ is large, more  competition exists in the top-$k$ ranking among POIs. Therefore, on average, fewer users can be included in the BR$k$NNs of each POI, incurring the results with less influence and lower accuracy.

\subsection{Effectiveness of Pruning Techniques}
\label{exp:effective-technique}

Finally, to determine whether evaluating socially relevant users and textually relevant users separately can achieve better pruning than evaluating all users together with combined score bounds~\cite{zhao2017icde}, we compare the numbers of unpruned users (i.e.,  candidates of BR$k$NN for verification) obtained by the two strategies. Fig.~\ref{fig:pruning-case}(a) presents the results on LasVegas. We observe that the number of unpruned users with separate pruning (performed by AP) is much smaller than the number of unpruned users with combined pruning (performed by BA). Also, Fig.~\ref{fig:pruning-case}(b) shows that  the number of score bound computations in AP is only $5$--$10\%$ of those in BA, which means score bounds derived in AP are much tighter, and thus most of the unnecessary computation in BA can be avoided by AP. This machine independent comparison demonstrates the advantage of Observation 1 when solving our problem.

\begin{figure}[t]
\centering
\vspace{-3mm}
\subfigure[\# of unpruned users]{
 \raisebox{-1.8mm}{\includegraphics[width=1.15in]{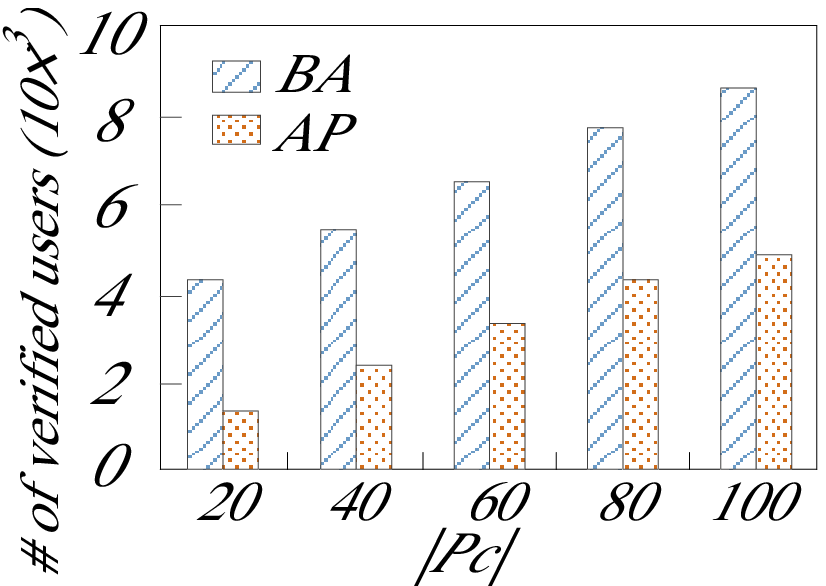}}
}
\hspace{0mm}
\subfigure[\# of bound computation]{
 \raisebox{-1.8mm}{\includegraphics[width=1.15in]{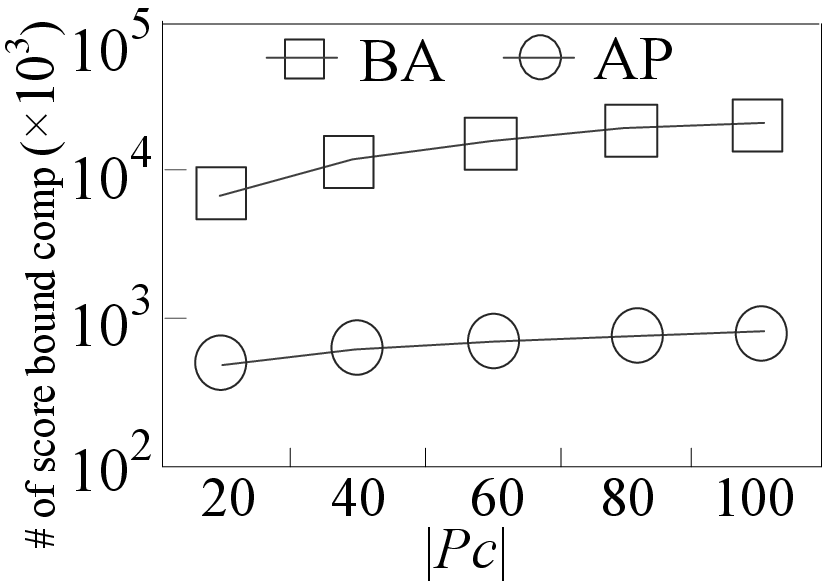}}
}
\vspace{-3.5mm}
\caption{Pruning performance comparison on the LasVegas dataset}
\label{fig:pruning-case}
\end{figure}

%% file: conclusion.tex
\section{Conclusions}
\label{sec:conclu}

We identify a new problem, called MaxInfBR$k$NN in geo-social networks. We prove that the problem is NP-hard, and develop a non-trivial baseline solution with theoretical guarantees by using  state-of-the-art techniques. To support efficient and scalable query processing, we present an efficient batch BR$k$NN processing framework, which encompasses effective POI selection policies to offer approximate and heuristic solutions to the problem. Extensive experiments on both real and synthetic datasets demonstrate the effectiveness and efficiency of our proposed methods. In the future, it is of interest to consider the problem in dynamic and distributed settings, where the data evolves over time.
\label{sec:con}

%% file: MaxInfBRkNN.bbl
\begin{thebibliography}{10}

\bibitem{tenindra2020tkde}
T.~Abeywickrama, M.~Cheema, and A.~Khan.
\newblock K-spin: Efficiently processing spatial keyword queries on road
  networks.
\newblock {\em TKDE}, 32(5):983--997, 2020.

\bibitem{ahuja2015sstd}
R.~Ahuja, N.~Armenatzoglou, D.~Papadias, and G.~Fakas.
\newblock Geo-social keyword search.
\newblock In {\em SSTD}, pages 431--450, 2015.

\bibitem{akiba2014alenex}
T.~Akiba, Y.~Iwata, K.~Kawarabayashi, and Y.~Kawata.
\newblock Fast shortest-path distance queries on road networks by pruned
  highway labeling.
\newblock In {\em ALENEX}, pages 147--154, 2014.

\bibitem{armenatzoglou2015vldbj}
N.~Armenatzoglou, R.~Ahuja, and D.~Papadias.
\newblock Geo-social ranking: Functions and query processing.
\newblock {\em VLDBJ}, 24(6):783--799, 2015.

\bibitem{bian2020vldb}
S.~Bian, Q.~Guo, S.~Wang, and J.~Yu.
\newblock Efficient algorithms for budgeted influence maximization on massive
  social networks.
\newblock {\em PVLDB}, 13(9):1498--1510, 2020.

\bibitem{borgs2014soda}
C.~Borgs, M.~Brautbar, J.~T. Chayes, and B.~Lucier.
\newblock Maximizing social influence in nearly optimal time.
\newblock In {\em SODA}, pages 946--957, 2014.

\bibitem{chen2010kdd}
W.~Chen, C.~Wang, and Y.~Wang.
\newblock Scalable influence maximization for prevalent viral marketing in
  large-scale social networks.
\newblock In {\em KDD}, page 1029–1038, 2010.

\bibitem{cheng2014sigir}
S.~Cheng, H.~Shen, J.~Huang, W.~Chen, and X.~Cheng.
\newblock Imrank: influence maximization via finding self-consistent ranking.
\newblock In {\em SIGIR}, pages 475--484, 2014.

\bibitem{choudhury2018vldbj}
F.~Choudhury, J.~Culpepper, Z.~Bao, and T.~Sellis.
\newblock Finding the optimal location and keywords in obstructed and
  unobstructed space.
\newblock {\em VLDBJ}, 27(4):445--470, 2018.

\bibitem{choudhury2016vldb}
F.~Choudhury, J.~Culpepper, T.~Sellis, and X.~Cao.
\newblock Maximizing bichromatic reverse spatial and textual $k$ nearest
  neighbor queries.
\newblock {\em PVLDB}, 9(6):1414--1417, 2016.

\bibitem{domingos2001kdd}
P.~Domingos and M.~Richardson.
\newblock Mining the network value of customers.
\newblock In {\em KDD}, pages 57--66, 2001.

\bibitem{feige1998acmj}
U.~Feige.
\newblock A threshold of ln n for approximating set cover.
\newblock {\em J. ACM}, 20(3):634–652, 1998.

\bibitem{gkorgkas2015sstd}
O.~Gkorgkas, A.~Vlachou, C.~Doulkeridis, and K.~Nørvag.
\newblock Maximizing influence of spatio-textual objects based on keyword
  selection.
\newblock In {\em SSTD}, pages 413--430, 2015.

\bibitem{guo2019icde}
F.~Guo, Y.~Yuan, G.~Wang, L.~Chen, X.~Lian, and Z.~Wang.
\newblock Cohesive group nearest neighbor queries over road-social networks.
\newblock In {\em ICDE}, pages 434--445, 2019.

\bibitem{guo2021icde}
F.~Guo, Y.~Yuan, G.~Wang, X.~Zhao, and H.Sun.
\newblock Multi-attributed community search in road-social networks.
\newblock In {\em ICDE}, pages 109--120, 2021.

\bibitem{guo2020sigmod}
Q.~Guo, S.~Wang, Z.Wei, and M.Chen.
\newblock Influence maximization revisited: Efficient reverse reachable set
  generation with bound tightened.
\newblock In {\em SIGMOD}, pages 62167--2181, 2020.

\bibitem{huang2020vldbj}
K.~Huang, J.~Tang, K.~Han, X.~Xiao, W.~Chen, A.~Sun, X.~Tang, and A.~Lim.
\newblock Efficient approximation algorithms for adaptive influence
  maximization.
\newblock {\em VLDBJ}, 29(6):1385--1406, 2020.

\bibitem{hung2014tsas}
H.~Hung, D.~Yang, and W.~Lee.
\newblock Social influence-aware reverse nearest neighbor search.
\newblock {\em ACM Trans. Spatial Algorithms Syst}, 2(3):1--35, 2016.

\bibitem{jin2020dasfaa}
P.~Jin, Y.~Gao, C.~Lu, and J.~Zhao.
\newblock Efficient group processing for multiple reverse top-k geo-social
  keyword queries.
\newblock In {\em DASFAA}, pages 279--287, 2020.

\bibitem{jin2021ickg}
P.~Jin, Z.~Liu, and Y.~Xiao.
\newblock Discovering the most influential geo-social object using location
  based social network data.
\newblock In {\em ICKG}, pages 607--614, 2021.

\bibitem{kempe2003kdd}
D.~Kempe, J.~Kleinberg, and E.~Tardos.
\newblock Maximizing the spread of influence through a social network.
\newblock In {\em KDD}, pages 137--146, 2003.

\bibitem{mohammad2004vldb}
M.~Kolahdouzan and C.~Shahabi.
\newblock Voronoi-based k nearest neighbor search for spatial network
  databases.
\newblock In {\em VLDB}, page 840–851, 2004.

\bibitem{ks2000sigmod}
F.~Korn and S.~Muthukrishnan.
\newblock Influence sets based on reverse nearest neighbor queries.
\newblock In {\em SIGMOD}, pages 201--212, 2000.

\bibitem{li2014sigmod}
G.~Li, S.~Chen, J.~Feng, K.~Tan, and W.~Li.
\newblock Efficient location-aware influence maximization.
\newblock In {\em SIGMOD}, pages 87--98, 2014.

\bibitem{yuchen2015vldb}
Y.~Li, D.~Zhang, and K.~Tan.
\newblock Real-time targeted influence maximization for online advertisements.
\newblock In {\em PVLDB}, pages 1070--1081, 2015.

\bibitem{lu2011sigmod}
J.~Lu, Y.~Lu, and G.~Cong.
\newblock Reverse spatial and textual $k$ nearest neighbor search.
\newblock In {\em SIGMOD}, pages 349--360, 2011.

\bibitem{lu2014acmtrans}
Y.~Lu, J.~Lu, G.~Cong, W.~Wu, and C.~Shahabi.
\newblock Efficient algorithms and cost models for reverse spatial-keyword
  k-nearest neighbor search.
\newblock {\em ACM Trans.Database Syst.}, 39(2):13:1--13:46, 2014.

\bibitem{luo2018dasfaa}
H.~Luo, F.~Choudhury, J.~Culpepper, Z.~Bao, and B.~Zhang.
\newblock Maxbrknn queries for streaming geo-data.
\newblock In {\em DASFAA}, pages 647--664, 2018.

\bibitem{richardson2002kdd}
M.~Richardson and P.~Domingos.
\newblock Mining knowledge-sharing sites for viral marketing.
\newblock In {\em KDD}, pages 61--70, 2002.

\bibitem{rocha2012edbt}
J.~Rocha-Junior and K.~Norvag.
\newblock Top-$k$ spatial keyword queries on road networks.
\newblock In {\em EDBT}, pages 168--179, 2012.

\bibitem{salton1988inf}
G.~Salton and C.~Buckley.
\newblock Term-weighting approaches in automatic text retrieval.
\newblock {\em Information Processing and Management}, 24(5):513--523, 1988.

\bibitem{rahat2018adc}
M.~A. T.~Rahat, A.~Arman.
\newblock Maximizing reverse k-nearest neighbors for trajectories.
\newblock In {\em ADC}, pages 262--274, 2018.

\bibitem{tang2018sigmod}
J.~Tang, X.~Tang, X.~Xiao, and J.~Yuan.
\newblock Online processing algorithms for influence maximization.
\newblock In {\em SIGMOD}, pages 991--1005, 2018.

\bibitem{tang2017asonam}
J.~Tang, X.~Tang, and J.~Yuan.
\newblock Influence maximization meets efficiency and effectiveness: A
  hop-based approach.
\newblock In {\em ASONAM}, pages 64--71, 2017.

\bibitem{tang2018snam}
J.~Tang, X.~Tang, and J.~Yuan.
\newblock An efficient and effective hop-based approach for influence
  maximization in social networks.
\newblock {\em Soc. Netw. Anal. Min.}, 8(1):10--560, 2018.

\bibitem{tang2015sigmod}
Y.~Tang, Y.~Shi, and X.~Xiao.
\newblock Influence maximization in near-linear time: A martingale approach.
\newblock In {\em SIGMOD}, pages 1539--1554, 2015.

\bibitem{tang2014sigmod}
Y.~Tang, X.~Xiao, and Y.~Shi.
\newblock Influence maximization: near-optimal time complexity meets practical
  efficiency.
\newblock In {\em SIGMOD}, pages 75--86, 2014.

\bibitem{sibo2016vldb}
S.~Wang, X.~Xiao, Y.~Yang, and W.~Lin.
\newblock Effective indexing for approximate constrained shortest path queries
  on large road networks.
\newblock {\em PVLDB}, 10(2):61--72, 2016.

\bibitem{wang2017tkde}
X.~Wang, Y.~Zhang, W.~Zhang, and X.~Lin.
\newblock Efficient distance-aware influence maximization in geosocial
  networks.
\newblock {\em TKDE}, 29(3):599¨C612, 2017.

\bibitem{raymond2009vldb}
R.~Wong, M.~Ozsu, P.~Yu, A.~Fu, and L.~Liu.
\newblock Efficient method for maximizing bichromatic reverse nearest neighbor.
\newblock {\em PVLDB}, 2(1):1126--1137, 2009.

\bibitem{wu2012mdm}
D.~Wu, Y.~Li, B.~Choi, and J.~Xu.
\newblock Social-aware top-$k$ spatial keyword search.
\newblock In {\em MDM}, pages 235--244, 2014.

\bibitem{wu2012tkde}
D.~Wu, M.~Yiu, G.~Cong, and C.~S. Jensen.
\newblock Joint top-$k$ spatial keyword query processing.
\newblock {\em TKDE}, 24(10):1889--1903, 2012.

\bibitem{yiu2006tkde}
M.~L. Yiu, D.~Papadias, N.~Mamoulis, and Y.~Tao.
\newblock Reverse nearest neighbors in large graphs.
\newblock {\em TKDE}, 18(4):540--553, 2006.

\bibitem{zhao2018icde}
J.~Zhao, Y.~Gao, G.~Chen, and R.~Chen.
\newblock Why-not questions on top-$k$ geo-social keyword queries in road
  networks.
\newblock In {\em ICDE}, pages 965--976, 2018.

\bibitem{zhao2017icde}
J.~Zhao, Y.~Gao, G.~Chen, C.~S. Jensen, R.~Chen, and D.~Cai.
\newblock Reverse top-$k$ geo-social keyword queries in road networks.
\newblock In {\em ICDE}, pages 387--398, 2017.

\bibitem{zhong2015tkde}
R.~Zhong, K.~T. G.~Li, L.~Zhou, and Z.~Gong.
\newblock G-tree: An efficient and scalable index for spatial search on road
  networks.
\newblock {\em TKDE}, 27(8):2175--2189, 2015.

\end{thebibliography}
